\documentclass[a4paper,onecolumn,11pt]{quantumarticle}
\pdfoutput=1
\ExplSyntaxOn
\ExplSyntaxOff
\usepackage[utf8]{inputenc}
\usepackage[english]{babel}
\usepackage[T1]{fontenc}
\usepackage{amsmath}
\usepackage{hyperref}
\usepackage{amssymb,amsfonts,amsthm}
\usepackage{mathtools}
\usepackage{graphicx}
\usepackage{tikz}
\usepackage{booktabs}
\theoremstyle{plain}
\newtheorem{theorem}{Theorem}
\newtheorem{lemma}{Lemma}
\newtheorem{proposition}{Proposition}
\newtheorem{corollary}{Corollary}
\theoremstyle{definition}
\newtheorem{definition}{Definition}
\newtheorem{remark}{Remark}

\DeclareMathOperator{\cone}{cone}
\newcommand{\Mf}{\mathbb{M}_{\mathbb{F}}}
\newcommand{\F}{\mathbb{F}}

\begin{document}
\title{Witness robustness: An operational quantifier of measurement resources via free state discrimination}
\author{Jingsong Ao}
\author{Aby Philip}
\author{Alexander Streltsov}
\affiliation{Institute of Fundamental Technological Research, Polish Academy of Sciences, Pawi\'{n}skiego 5B, 02-106 Warsaw, Poland}

\maketitle

\begin{abstract}
We introduce the \emph{witness robustness} of quantum measurements, a resource quantifier whose admissible noise consists of tuples of free-state witnesses rather than physical measurements. We establish its operational interpretation: it quantifies the maximal advantage that a measurement can provide over free measurements in discriminating an ensemble composed entirely of free states. Unlike the standard and generalized robustnesses, the witness robustness is not faithful in general, reflecting the fact that a resourceful measurement need not be useful when only free states can be prepared. We identify conditions under which faithfulness is recovered and show that, in resource theories admitting a resource-destroying map, the witness robustness vanishes for every measurement. We also establish fundamental properties, including convexity and monotonicity. Finally, we derive analytical results for projective measurements in single-qubit magic and for binary pure-state projective measurements in the two-qubit PPT entanglement theory.
\end{abstract}

% ===========================================================================
\section{Introduction}
\label{sec:intro}
Quantum resource theories provide a general framework for describing the physical properties that enable quantum advantages~\cite{ChitambarGour2019}. A resource theory specifies which states, channels, and measurements are free, with all other objects regarded as resourceful. Commonly considered resource theories include entanglement~\cite{HorodeckiReview2009}, magic~\cite{Veitch2014,Howard2017}, coherence~\cite{Baumgratz2014,Streltsov2017}, and asymmetry~\cite{GourSpekkens2008}. Once these free sets have been identified, a central task is to quantify how much resource a given object, such as a state, channel, or measurement, contains. Resource measures address this question by assigning a numerical value to each object while respecting the transformations allowed by the theory~\cite{ChitambarGour2019}.

Robustness is a class of resource measures that quantifies the resources of states, channels, and measurements in a unified manner. By definition, it measures the minimum amount of noise that must be admixed with an object in order to make it free. If the noise is restricted to free objects, one obtains the standard robustness~\cite{VidalTarrach1999}; if it may be any physical object of the same type, one obtains the generalized robustness~\cite{Steiner2003,HarrowNielsen2003}. Both variants have been formulated for general convex resource theories~\cite{Regula2018,TakagiRegula2019}.

In this work, we study the \emph{witness robustness} of measurements, for which the noise is a tuple of free-state witnesses. These witnesses are observables whose expectation values are nonnegative on every free state. This goes beyond the usual restriction where the added noise is of the same physical type. Note that every operator of a measurement must be positive semidefinite and therefore has a nonnegative expectation value on every state, whereas the added noise is a free-state witness that is required to be nonnegative only on free states. We show that witness robustness has an operational interpretation in terms of a quantum state discrimination task.

The connection between quantum state discrimination, introduced in the works of Helstrom and Holevo~\cite{Helstrom1976,Holevo1973}, and robustness has already been established. The operational interpretation of the generalized robustness of measurements was introduced by Skrzypczyk and Linden~\cite{SkrzypczykLinden2019} and further developed into a general theoretical framework in~\cite{TakagiRegula2019,OszmaniecBiswas2019,Uola2019}. Simply put, this framework uses a state-discrimination task to quantify the amount of resource contained in a measurement. However, the states involved in such tasks, so far, are unrestricted and hence allowed to be resourceful states. If the state ensemble is required to consist entirely of free states, the corresponding robustness is no longer the generalized robustness, but the \emph{witness robustness} introduced in this work. This quantity takes into account not only the set of free measurements but also the set of free states, and quantifies the maximal advantage that a measurement can provide over free measurements in discriminating a free-state ensemble.

In general, the witness robustness is not faithful. For resource theories where witness robustness is not faithful, a resourceful measurement does not necessarily provide an advantage in the discrimination of free-state ensembles. In this paper, we show that witness robustness is faithful in resource theories where the set of free states has a non-empty interior~\cite{ao2026geometricorigins}. We also show that witness robustness is unfaithful in resource theories admitting a resource-destroying map~\cite{Liu2017}. This distinguishes it from both the generalized robustness and the standard robustness, which are always faithful~\cite{TakagiRegula2019,SkrzypczykLinden2019,OszmaniecBiswas2019}. Operationally, the faithfulness of the generalized robustness means that every resourceful measurement provides an advantage in the discrimination of some state ensemble, although the ensemble exhibiting this advantage need not consist entirely of free states. When the preparation of resource states is unavailable, the witness robustness therefore provides a more appropriate resource measure for quantum measurements. 

Note that the condition for the faithfulness of witness robustness is the same as the condition for the existence of free-state discrimination gaps in a resource theory, derived by the authors in~\cite{ao2026geometricorigins}. A free-state discrimination gap in a resource theory is defined by the existence of a set of free states that are strictly less distinguishable when using only free measurements than if all possible measurements were allowed. 

This result is interesting from the perspective of quantum data hiding~\cite{PhysRevLett.86.5807,EggelingWerner2002,DiVincenzo2002,Hayden_2004,Aubrun_2015}. The authors previously proved that, in discrimination tasks only involving free states, neither a finite-dimensional quantum memory nor a quantum catalyst can asymptotically improve the discrimination success rate~\cite{ao2026geometricorigins,philip2025robustnessdatahiding}. Free-state ensembles therefore possess a distinctive advantage in cryptographic security, providing a practical motivation for studying the witness robustness. Moreover, the maximum witness robustness over all measurements characterizes the optimal limit of data hiding achievable within the resource theory defined by its free measurements, free states, and free channels.

This work is structured as follows. Section~\ref{sec:prelim} fixes the setting and recalls the standard and generalized robustness together with the dual-cone structure. Section~\ref{sec:def} defines the witness robustness and establishes a hierarchy of robustness. We establish the operational interpretation of the witness robustness in Section~\ref{sec:main}. Section~\ref{sec:properties} investigates several fundamental properties of the witness robustness, including monotonicity and convexity. In Section~\ref{sec:faithful}, we investigate the resource theories in which witness robustness is faithful. In Section~\ref{sec:rdm}, we show that witness robustness vanishes for every measurement in resource theories that admit resource-destroying maps. The latter case implies that any free-state ensemble can be optimally discriminated by a free measurement. We also derive analytical results for two resource theories in Section~\ref{sec:examples}. The first concerns the witness robustness of projective measurements in the resource theory of single-qubit magic. The second concerns, in the two-qubit PPT entanglement theory, the binary measurement consisting of the projector onto a pure state and its orthogonal complement. Finally, in Section~\ref{sec:conclusion}, we conclude with limitations and open problems.

\section{Robustness measures and the witness cone}
\label{sec:prelim}

Before introducing the witness robustness, we recall two previously defined robustness measures. Each gauges resourcefulness of an object by asking how much noise must be mixed in before the object becomes free. The two differ only in which noise is admissible. We state them in parallel for states, channels, and measurements, since the constructions are identical. 

\subsection{Setting and notation}
\label{sec:setting}
To begin, we adopt the language of general probabilistic theories, which places states, measurements, and channels on a common geometric footing and recovers quantum theory as a special case~\cite{Hardy2001,Barrett2007,Plavala2023}. 

In a general probabilistic theory, a system is described by a finite-dimensional real vector space $V$. Its states form a compact convex set $\Omega(V)\subset V$, and we write $\sigma\in\Omega(V)$ for a state. The linear functionals on $V$ form the dual space $V^*$, paired with $V$ through $\langle\cdot,\cdot\rangle:V^*\times V\to\mathbb{R}$. An \emph{effect} is an element $e\in V^*$ with $0\le\langle e,\sigma\rangle\le1$ for all $\sigma\in\Omega(V)$, read as an outcome probability, and the unit $U\in V^*$ is the deterministic effect, $\langle U,\sigma\rangle=1$ for every $\sigma\in\Omega(V)$.

A measurement is a tuple $\mathcal{M}=\{M_i\}_i$ of effects summing to the unit, $\sum_i M_i=U$, so that $\{\langle M_i,\sigma\rangle\}_i$ is a probability distribution for every state; $\mathbb{M}$ denotes the set of all measurements, and we write $\mathcal{M}=\{M_i\}$ for a target measurement and $\mathcal{F}=\{F_i\}$ for a free one. Channels are the admissible transformations between systems, collected in $\mathbb{T}(V,V')$. A convex resource theory is specified by its free objects: a closed convex set of free states $\F\subseteq\Omega(V)$, free measurements $\Mf\subseteq\mathbb{M}$, and free channels $\mathbb{O}_{\F}\subseteq\mathbb{T}(V,V')$ when transformations are considered. 

To recover quantum mechanics, one can make the following choices: $V=\mathrm{Herm}(\mathbb{H})$ is the space of Hermitian operators on a finite-dimensional Hilbert space, $\Omega(V)$ is the set of density operators. The Hilbert--Schmidt inner product renders $V$ self-dual, so that $V^*$ is identified with $V$ and $\langle W,\sigma\rangle=\mathrm{Tr}(W\sigma)$. Effects are then operators $0\le M_i\le\mathbb{I}$, the unit effect is $U=\mathbb{I}$, and a measurement is a POVM. Now, we have everything to begin discussing robustness measures. 

\subsection{Standard robustness}

The standard robustness originates in the robustness of entanglement of Vidal and Tarrach~\cite{VidalTarrach1999}; it has since been formulated for general convex resource theories~\cite{Regula2018} and for quantum measurements~\cite{PhysRevLett.127.060402}. It is the most conservative robustness measure as it only admits objects that are already free as noise. The standard robustness, informally, the smallest weight of free noise such that mixing the target object with it yields a free object. 

For a state $\rho\in\Omega(V)$ with $\F\subseteq\Omega(V)$ closed and convex~\cite{Regula2018},
\begin{equation}
R_{\F}^{\mathrm{std}}(\rho) := \inf\left\{ r\ge0 \;\middle|\;
\frac{\rho+r\sigma}{1+r}\in\F,\ \sigma\in\F \right\}.
\end{equation}

For a channel $\Lambda\in\mathbb{T}(V,V')$ with $\mathbb{O}_{\F}$ closed and convex~\cite{PhysRevLett.127.060402},
\begin{equation}
R_{\mathbb{O}_{\F}}^{\mathrm{std}}(\Lambda) := \inf\left\{ r\ge0 \;\middle|\;
\frac{\Lambda+r\Theta}{1+r}\in\mathbb{O}_{\F},\ \Theta\in\mathbb{O}_{\F} \right\}.
\end{equation}

For a measurement $\mathcal{M}=\{M_i\}_i\in\mathbb{M}$ with $\Mf$ closed and convex~\cite{PhysRevLett.127.060402},
\begin{equation}
R_{\Mf}^{\mathrm{std}}(\mathcal{M}) := \inf\left\{ r\ge0 \;\middle|\;
\frac{\mathcal{M}+r\mathcal{N}}{1+r}\in\Mf,\ \mathcal{N}\in\Mf \right\}.
\end{equation}

\subsection{Generalized robustness}

The generalized robustness was introduced for entanglement by Steiner~\cite{Steiner2003}. It serves as a canonical monotone in general convex resource theories~\cite{BrandaoGour2015,Regula2018,TakagiRegula2019}, including those of quantum measurements~\cite{OszmaniecBiswas2019,Skrzypczyk2019}. It keeps the same template but allows the noise to range over the full physical set of objects of the relevant type.

For a state~\cite{BrandaoGour2015},
\begin{equation}
R_{\F}(\rho) := \inf\left\{ r\ge0 \;\middle|\;
\frac{\rho+r\sigma}{1+r}\in\F,\ \sigma\in\Omega(V) \right\}.
\end{equation}

For a channel~\cite{TakagiRegula2019},
\begin{equation}
R_{\mathbb{O}_{\F}}(\Lambda) := \inf\left\{ r\ge0 \;\middle|\;
\frac{\Lambda+r\Theta}{1+r}\in\mathbb{O}_{\F},\ \Theta\in\mathbb{T}(V,V') \right\}.
\end{equation}

For a measurement~\cite{OszmaniecBiswas2019,Skrzypczyk2019},
\begin{equation}
R_{\Mf}(\mathcal{M}) := \inf\left\{ r\ge0 \;\middle|\;
\frac{\mathcal{M}+r\mathcal{N}}{1+r}\in\Mf,\ \mathcal{N}\in\mathbb{M} \right\}.
\end{equation}

The two variants differ only in which noise counts as admissible, and enlarging the noise set can only make the free object easier to reach, so $R^{\mathrm{std}}\ge R$ always. In both cases, however, the noise remains an object of the same kind as the target---a measurement is perturbed by another measurement. The witness robustness of Sec.~\ref{sec:def} breaks with this: it lets the noise be a tuple of free-state witnesses, observables drawn from the dual cone introduced next.

\subsection{Cones and dual cones}

We define the conical hull of the free states and its dual cone, 
\begin{align}
\cone(\F) &:= \{ \lambda\sigma \mid \lambda\ge0,\ \sigma\in\F \},\\
\cone^*(\F) &:= \left\{ W\in V^* \,\middle|\,
   \langle W,\sigma\rangle\ge 0,\ \forall\sigma\in\cone(\F) \right\}\nonumber\\
&= \left\{ W\in V^* \,\middle|\,
   \langle W,\sigma\rangle\ge 0,\ \forall\sigma\in\F \right\}.
\end{align}
where $V^*$ is the dual space of $V$. This dual cone is exactly the set of all witnesses of free states. If $\F\subseteq\Omega(V)$, the cones satisfy the inclusion relation~\cite{boyd2004convex}
\begin{equation}
\cone(\F)\subseteq\cone(\Omega(V))
\implies \cone^*(\Omega(V))\subseteq\cone^*(\F),
\end{equation}
that is, \emph{the smaller the cone, the larger its dual}.

A witness $W\in\cone^*(\F)$ is simply an observable that assigns a non-negative expectation value to every free state. The dual cone, $\cone^*(\F)$, is the set of all such witnesses. The inclusion above then has an intuitive meaning: shrinking the free state set enlarges the family of observables that stay non-negative on it, so a more restrictive resource theory comes equipped with a richer supply of witnesses. We will exploit this enlarged set of admissible noise in the next section.

\section{Witness robustness of a measurement}
\label{sec:def}

We now define our quantity of interest, witness robustness.
\begin{definition}[Witness robustness]
\label{def:wr}
Let $\mathcal{M}=\{M_i\}_i\in\mathbb{M}$, and let $\Mf\subseteq\mathbb{M}$ and $\F\subseteq\Omega(V)$ be closed and convex. The witness robustness is
\begin{equation}
\label{eq:wr}
R_{\Mf}^{\F}(\mathcal{M}) := \inf\left\{ r\ge0 \;\middle|\;
\frac{\mathcal{M}+r\mathcal{W}}{1+r}\in\Mf,\ \mathcal{W}=\{W_i\},\ W_i\in\cone^*(\F) \right\}.
\end{equation}
\end{definition}

The two robustness measures of the previous section both admix as noise an object of the same kind as the target: a measurement is perturbed by another measurement. The witness robustness makes a single change: the noise components $W_i$ need no longer be effects, but are drawn from the larger set of free-state witnesses $W_i\in\cone^*(\F)$ that are observables in $V^*$ such that $\langle W_i,\sigma\rangle\ge 0$ for every free state $\sigma$. The tuple $\mathcal{W}$ is then not required to be a measurement, and geometrically each noise component is relaxed from the set of effects to the much larger dual cone.

\begin{figure}[t]
  \centering
  \begin{tikzpicture}[scale=1.0]
  \draw[->,blue!45!black,thick] (0,-2.7) -- (-4.7,1.53);
  \draw[->,blue!45!black,thick] (0,-2.7) -- ( 4.7,1.53);
  \fill[blue!45!black] (0,-2.7) circle (1.3pt);
  \node[font=\footnotesize,blue!45!black,anchor=north] at (0,-2.78) {$0$};
  \node[font=\footnotesize,blue!55!black] at (0,1.9) {$[\cone^*(\F)]^{\times N}$};
  \fill[blue!12] (0,0) ellipse (2.3 and 1.4);
  \fill[blue!26] (0,0) ellipse (1.0 and 0.65);
  \node[font=\footnotesize] at (0,0) {$\Mf$};
  \node[font=\footnotesize] at (0,1.05) {$\mathbb{M}$};
  \fill (1.55,-0.5) circle (1.7pt);
  \node[font=\footnotesize,above right] at (1.55,-0.5) {$\mathcal{M}$};
  \fill (-2.6,1.5) circle (1.7pt);
  \node[font=\footnotesize,above] at (-2.6,1.5) {$\mathcal{W}$};
  \draw[dashed,thick] (1.55,-0.5) -- (-2.6,1.5);
  \fill[red!75!black] (0.947,-0.209) circle (1.5pt);
  \node[font=\footnotesize,red!70!black,anchor=north] at (0.9,-0.42)
      {$\tfrac{\mathcal{M}+r\mathcal{W}}{1+r}$};
  \end{tikzpicture}
  \caption{Witness robustness as a relaxation of the admixed noise. All sets live in the space of $N$-tuples of observables: the free measurements $\Mf$ sit inside all measurements $\mathbb{M}$, both bounded convex sets, which in turn sit inside the \emph{unbounded} product cone $[\cone^*(\F)]^{\times N}$ of witness tuples---tuples $\mathcal{W}=\{W_i\}$ whose every component is a free-state witness---with apex at the zero tuple $0$. The witness robustness mixes the resourceful target $\mathcal{M}$ with such a tuple $\mathcal{W}$---\emph{not} with another measurement---until the convex combination $(\mathcal{M}+r\mathcal{W})/(1+r)$ first meets $\Mf$; the least weight $r$ is $R^{\F}_{\Mf}(\mathcal{M})$. Restricting the noise instead to $\mathbb{M}$ or to $\Mf$ recovers the generalized and standard robustness, so that $R^{\F}_{\Mf}\le R_{\Mf}\le R^{\mathrm{std}}_{\Mf}$.}
  \label{fig:schematic}
\end{figure}
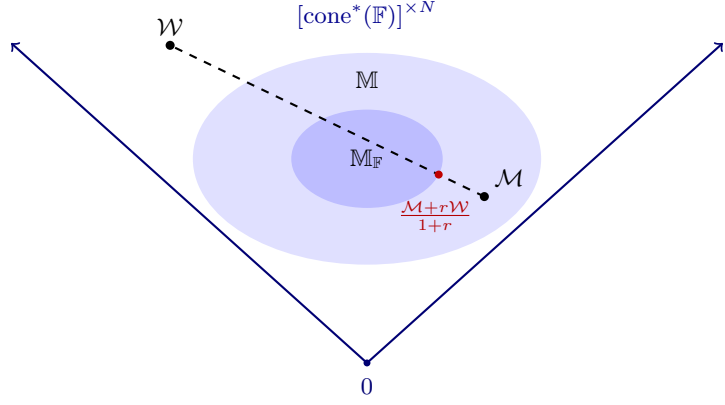

The three noise sets are nested: free measurements lie inside the set of all measurements, which in turn lies inside the product cone $[\cone^*(\F)]^{\times N}$ of witness tuples, see Fig.~\ref{fig:schematic}. Hence the witness robustness satisfies
\begin{equation}\label{eq:hierarchy}
    R_{\Mf}^{\F}(\mathcal{M}) \le R_{\Mf}(\mathcal{M}) \le R_{\Mf}^{\mathrm{std}}(\mathcal{M}).
\end{equation}
The generalized robustness admits a discrimination-task characterization with arbitrary state ensembles~\cite{OszmaniecBiswas2019,SkrzypczykLinden2019,TakagiRegula2019}. In the next section, we present Theorem~\ref{thm:main} which provides the analogue for the witness robustness, the smallest of the three, with the ensembles restricted to free states.

\section{Operational interpretation}
\label{sec:main}

We now turn to the operational interpretation of witness robustness. We begin by considering a discrimination game played with free states alone: a referee draws from a free-state ensemble by picking a label $i$ with probability $p_i$ and preparing the free state $\sigma_i\in\F$ for the player; the player applies a fixed measurement and guesses the label from its outcome. Formally, a \emph{free-state ensemble} is a collection $A_{\F}=\{p_i,\sigma_i\}_{i=1}^{N}$ with $\sigma_i\in\F$, $p_i\ge0$, $\sum_i p_i=1$, and the success probability of a measurement $\mathcal{M}=\{M_i\}_{i=1}^N$ is
\begin{equation}
P_{\text{succ}}(A_{\F},\mathcal{M}):=\sum_{i=1}^{N}p_i\,\langle M_i,\sigma_i\rangle.
\end{equation}
The theorem below shows that the maximal advantage of $\mathcal{M}$ over the best free measurement in this task is controlled exactly by the witness robustness, mirroring the characterization known for the generalized robustness~\cite{TakagiRegula2019,OszmaniecBiswas2019} in the restricted free-state setting.

\begin{theorem}[Operational interpretation]
\label{thm:main}
Fix $N\in\mathbb{N}$. Let $\F\subseteq\Omega(V)$ be a non-empty closed convex set of free states, and let $\Mf$ be a closed convex set of $N$-outcome free measurements. Then, for every $N$-outcome measurement $\mathcal{M}$,
\begin{equation}
\label{eq:main}
\max_{A_{\F}} \frac{P_{\text{succ}}(A_{\F},\mathcal{M})}{\max_{\mathcal{F}\in\Mf} P_{\text{succ}}(A_{\F},\mathcal{F})} = 1 + R_{\Mf}^{\F}(\mathcal{M}).
\end{equation}
\end{theorem}

\begin{proof}
To begin, define $\omega = \{\omega_i\}$, where $\omega_i = p_i \sigma_i$ with $\sigma_i \in \mathbb{F}$. This implies $\omega_i \in \text{cone}(\mathbb{F})$.
We define the set $\Omega_{\mathbb{F}}$ as:
\begin{equation}
\Omega_{\mathbb{F}} := \left\{ \omega = \{\omega_i\} \;\middle|\; \omega_i \in \text{cone}(\mathbb{F}), \quad \sum_i \langle U, \omega_i \rangle = 1 \right\}
\end{equation}
Since $\mathbb{F}$ is compact (a closed, bounded subset of the state space) and convex, so is $\Omega_{\mathbb{F}}$; likewise $\mathbb{M}_{\mathbb{F}}$ is compact (its effects obey $0 \le F_i \le U$) and convex.\newline

\noindent Given a target measurement, $\mathcal{M}$, and for any $\lambda\ge0$, $\omega \in \Omega_{\mathbb{F}}$, and $\mathcal{F} \in \mathbb{M}_{\mathbb{F}}$, we define the following function:
\begin{equation}
f(\omega, \mathcal{F}; \lambda) := \langle \mathcal{M}, \omega \rangle - \lambda \langle \mathcal{F}, \omega \rangle
\end{equation}
We define the value function $V(\lambda)$ as:
\begin{equation}
V(\lambda) := \max_{\omega \in \Omega_{\mathbb{F}}} \min_{\mathcal{F} \in \mathbb{M}_{\mathbb{F}}} f(\omega, \mathcal{F}; \lambda)
\end{equation}
As $\Omega_{\mathbb{F}}$ and $\mathbb{M}_{\mathbb{F}}$ are closed and bounded in finite dimension, they are compact. Note that, $f(\omega, \mathcal{F}; \lambda)$ is continuous, concave in $\omega$ for
every fixed $\mathcal{F}$, and convex in $\mathcal{F}$ for every fixed
$\omega$. Using Sion's minimax theorem~\cite{Sion1958}, we can swap the max and min operators:
\begin{equation}\label{eqn:min_max_theorem_1}
    \max_{\omega \in \Omega_{\mathbb{F}}} \min_{\mathcal{F} \in \mathbb{M}_{\mathbb{F}}} f(\omega, \mathcal{F}; \lambda) = \min_{\mathcal{F} \in \mathbb{M}_{\mathbb{F}}} \max_{\omega \in \Omega_{\mathbb{F}}} f(\omega, \mathcal{F}; \lambda)
\end{equation}

\vspace{1em}
\noindent Recall that we are looking to quantify advantage provided by the target measurement. Hence, we want to find the minimal $\lambda$ to make $V(\lambda) \le 0$. Then,
\begin{align}
V(\lambda)= \max_{\omega \in \Omega_{\mathbb{F}}} \min_{\mathcal{F} \in \mathbb{M}_{\mathbb{F}}} f(\omega, \mathcal{F}; \lambda) &\le 0 
\end{align}
Since the maximum of a function is non-positive, we have that 
\begin{align}
\forall \omega \in \Omega_{\mathbb{F}} : \quad \min_{\mathcal{F} \in \mathbb{M}_{\mathbb{F}}} f(\omega, \mathcal{F}; \lambda) &\le 0 
\end{align}
And, since $\lambda\ge 0$, minimizing over $\mathbb{M}_{\mathbb{F}}$ yields:
\begin{align}
\langle \mathcal{M}, \omega \rangle - \lambda \max_{\mathcal{F} \in \mathbb{M}_{\mathbb{F}}} \langle \mathcal{F}, \omega \rangle &\le 0 
\end{align}
By rearranging, we get
\begin{align}
\lambda &\ge \frac{\langle \mathcal{M}, \omega \rangle}{\max_{\mathcal{F} \in \mathbb{M}_{\mathbb{F}}} \langle \mathcal{F}, \omega \rangle}
\end{align}

\vspace{1em}
\noindent The above expression can be rewritten in terms of the probability of distinguishing the free-state ensemble $A_{\mathbb{F}}$ in the following fashion:
\begin{equation}
    \lambda \ge \frac{P_{\text{succ}}(A_{\mathbb{F}}, \mathcal{M})}{\max_{\mathcal{F} \in \mathbb{M}_{\mathbb{F}}} P_{\text{succ}}(A_{\mathbb{F}}, \mathcal{F})}.
\end{equation}
This inequality needs to hold for all $A_{\mathbb{F}}$, therefore the inequality must hold for the maximum over all $A_{\mathbb{F}}$ since $\Omega_{\mathbb{F}}$ is compact. Hence 
\begin{equation}
\lambda \geq \max_{A_{\mathbb{F}}} \frac{P_{\text{succ}}(A_{\mathbb{F}}, \mathcal{M})}{\max_{\mathcal{F} \in \mathbb{M}_{\mathbb{F}}} P_{\text{succ}}(A_{\mathbb{F}}. \mathcal{F})}
\end{equation}
Conversely, the displayed bound implies the preceding inequality for every $A_{\F}$, and hence $V(\lambda)\le0$; therefore, all the implications above are equivalences.
\vspace{1em}
Now, let us examine the right hand side of \eqref{eqn:min_max_theorem_1}.
\begin{align}
\min_{\mathcal{F} \in \mathbb{M}_{\mathbb{F}}} \max_{\omega \in \Omega_{\mathbb{F}}} f(\omega, \mathcal{F}; \lambda) &\le 0 
\end{align}
Since the minimum of a function is non-positive, we have that there exists a $\mathcal{F} \in \mathbb{M}_{\mathbb{F}}$ such that 
\begin{align}
    \max_{\omega \in \Omega_{\mathbb{F}}} f(\omega, \mathcal{F}; \lambda) \le 0
\end{align}
We expand the inner product using the ensemble $\{p_i, \sigma_i\}$ where $\sigma_i \in \mathbb{F}$:
\begin{align}
    &\max_{\{p_i, \sigma_i\}, \sigma_i \in \mathbb{F}} \sum_i p_i \langle M_i - \lambda F_i, \sigma_i \rangle \le 0 \\
    &\implies \max_{\sigma \in \mathbb{F},\, i} \langle M_i - \lambda F_i, \sigma \rangle \le 0
\end{align}
\noindent Flipping the sign, we get
\begin{align}
    \exists \mathcal{F} \in \mathbb{M}_{\mathbb{F}} &: \quad \forall i, \; \forall \sigma \in \mathbb{F},\,\, \langle \lambda F_i - M_i, \sigma \rangle \ge 0 \\
    \exists \mathcal{F} \in \mathbb{M}_{\mathbb{F}} &, \,\, \forall i, \quad \lambda F_i - M_i \in \text{cone}^*(\mathbb{F})
\end{align}

\vspace{1em}
\noindent The tuple $\{\lambda F_i - M_i\}_i$ acts as a tuple of free-state witnesses.

\vspace{1em}
\noindent Conversely, these cone conditions imply $f(\omega,\mathcal F;\lambda)\le0$ for every $\omega\in\Omega_{\F}$, and hence $\min_{\mathcal F\in\Mf}\max_{\omega\in\Omega_{\F}} f(\omega,\mathcal F;\lambda)\le0$; therefore, all the implications above are equivalences.

\vspace{1em}
\noindent Combining the two sides of Sion's identity, the admissible $\lambda$ coincide:
\begin{align}
    \left\{\lambda \;\middle|\;V(\lambda)\leq0\right\}&=
\left\{ \lambda \;\middle|\; \lambda \ge \max_{A_{\mathbb{F}}} \frac{P_{\text{succ}}(A_{\mathbb{F}}, \mathcal{M})}{\max_{\mathcal{F} \in \mathbb{M}_{\mathbb{F}}} P_{\text{succ}}(A_{\mathbb{F}}, \mathcal{F})} \right\}\\
&= \left\{ \lambda \;\middle|\; \forall i,\,\, \lambda F_i - M_i \in \text{cone}^*(\mathbb{F}),\, \{F_i\} \in \mathbb{M}_{\mathbb{F}}\right\}.
\end{align}
Taking infima on both sides,
\begin{equation}
\begin{split}
&\inf \left\{ \lambda \;\middle|\; \lambda \ge \max_{A_{\mathbb{F}}} \frac{P_{\text{succ}}(A_{\mathbb{F}}, \mathcal{M})}{\max_{\mathcal{F} \in \mathbb{M}_{\mathbb{F}}} P_{\text{succ}}(A_{\mathbb{F}}, \mathcal{F})} \right\} \\
&\qquad = \inf \left\{ \lambda \;\middle|\; \forall i,\,\, \lambda F_i - M_i \in \text{cone}^*(\mathbb{F}),\, \{F_i\} \in \mathbb{M}_{\mathbb{F}} \right\},
\end{split}
\label{eq:two_sides_theorem_1}
\end{equation}
where the left-hand side is exactly $\max_{A_{\mathbb{F}}} P_{\text{succ}}(A_{\mathbb{F}}, \mathcal{M}) / \max_{\mathcal{F} \in \mathbb{M}_{\mathbb{F}}} P_{\text{succ}}(A_{\mathbb{F}}, \mathcal{F})$. 

Note that by summing $\lambda F_i - M_i \in \text{cone}^*(\mathbb{F})$ over all $i$ and using $\sum_i F_i = \sum_i M_i = U$, we get
\begin{equation}
\sum_i (\lambda F_i - M_i) = (\lambda - 1) U \in \text{cone}^*(\mathbb{F}).
\end{equation}
By the definition of $\text{cone}^*(\mathbb{F})$, for any $\sigma \in \mathbb{F}$, $\langle (\lambda - 1) U, \sigma \rangle \ge 0$; since $\langle U, \sigma \rangle = 1$, this gives $\lambda \ge 1$. We accordingly set $r := \lambda - 1 \ge 0$.

We relate the right-hand side of \eqref{eq:two_sides_theorem_1} to the witness robustness
\begin{equation}
R_{\mathbb{M}_{\mathbb{F}}}^{\mathbb{F}}(\mathcal{M}) = \inf \left\{ r \ge0 \;\middle|\; \frac{\mathcal{M} + r\mathcal{W}}{1+r} \in \mathbb{M}_{\mathbb{F}},\ \mathcal{W} = \{W_i\},\ W_i \in \text{cone}^*(\mathbb{F}) \right\},
\end{equation}
We will now prove that
\begin{equation}\label{eq:inf}
\begin{split}
    &\inf \left\{ \lambda \ge 1 \;\middle|\; \forall i,\,\, \lambda F_i - M_i \in \text{cone}^*(\mathbb{F}),\, \{F_i\} \in \mathbb{M}_{\mathbb{F}} \right\} \\
    &\quad = \ 1 + \inf \left\{ r \ge 0 \;\middle|\; \exists \mathcal{W} = \{W_i\},\ W_i \in \text{cone}^*(\mathbb{F}),\ \frac{\mathcal{M} + r\mathcal{W}}{1+r} \in \mathbb{M}_{\mathbb{F}} \right\}.
\end{split}
\end{equation}
For $\lambda > 1$ (i.e.\ $r > 0$), set $\lambda = 1 + r$ and define
\begin{equation}
W_i := \frac{(1+r)F_i - M_i}{r} = \frac{\lambda F_i - M_i}{r}.
\end{equation}
Since $\text{cone}^*(\mathbb{F})$ is a cone and $r > 0$, we have $\lambda F_i - M_i \in \text{cone}^*(\mathbb{F}) \iff W_i \in \text{cone}^*(\mathbb{F})$. Rearranging,
\begin{equation}
F_i = \frac{r W_i + M_i}{1+r},
\end{equation}
so $\mathcal{F} = \{F_i\}_{i=1}^N \in \mathbb{M}_{\mathbb{F}}$ is equivalent to $\frac{\mathcal{M} + r\mathcal{W}}{1+r} \in \mathbb{M}_{\mathbb{F}}$. Hence the two sets
\begin{equation}
\left\{ \lambda > 1 \;\middle|\; \forall i,\,\, \lambda F_i - M_i \in \text{cone}^*(\mathbb{F}),\, \{F_i\} \in \mathbb{M}_{\mathbb{F}} \right\} 
\end{equation}
and
\begin{equation}
\left\{ r > 0 \;\middle|\; \frac{\mathcal{M} + r\mathcal{W}}{1+r} \in \mathbb{M}_{\mathbb{F}},\,\,\mathcal{W} = \{W_i\},\ W_i \in \text{cone}^*(\mathbb{F}) \right\}
\end{equation}
are in one-to-one correspondence via $\lambda = 1 + r$.

\vspace{0.5em}
\noindent For $\lambda = 1$ the condition is $\exists \mathcal{F} \in \mathbb{M}_{\mathbb{F}} : \forall i,\ F_i - M_i \in \text{cone}^*(\mathbb{F})$, which is \emph{more general} than the $r = 0$ condition
\begin{equation}
\frac{\mathcal{M} + 0 \cdot \mathcal{W}}{1+0} \in \mathbb{M}_{\mathbb{F}}, \quad \text{i.e.} \quad \mathcal{M} \in \mathbb{M}_{\mathbb{F}}, \quad M_i = F_i.
\end{equation}
So every $r = 0$ case gives $\lambda = 1$, but a $\lambda = 1$ case need not give $r = 0$. This mismatch does not affect the infima: both sets are upward closed (if $x$ is in the set, so is any $y \ge x$), so discarding the single smallest point $\lambda = 1$ ($r = 0$) leaves the infimum unchanged:
\begin{align}
&\inf \left\{ \lambda \ge 1 \;\middle|\; \forall i,\,\, \lambda F_i - M_i \in \text{cone}^*(\mathbb{F}),\, \{F_i\} \in \mathbb{M}_{\mathbb{F}} \right\} \notag\\
&= \inf \left\{ \lambda > 1 \;\middle|\; \forall i,\,\, \lambda F_i - M_i \in \text{cone}^*(\mathbb{F}),\, \{F_i\} \in \mathbb{M}_{\mathbb{F}} \right\}  \notag\\
&=1 + \inf  \left\{ r > 0 \;\middle|\; \frac{\mathcal{M} + r\mathcal{W}}{1+r} \in \mathbb{M}_{\mathbb{F}},\,\,\mathcal{W} = \{W_i\},\ W_i \in \text{cone}^*(\mathbb{F})  \right\}  \\
&=1 + \inf \left\{ r \ge 0 \;\middle|\; \frac{\mathcal{M} + r\mathcal{W}}{1+r} \in \mathbb{M}_{\mathbb{F}},\,\,\mathcal{W} = \{W_i\},\ W_i \in \text{cone}^*(\mathbb{F}) \right\}
\end{align}
The second equality follows from $\lambda = 1 + r$ whenever $\lambda>1$. Thus
\begin{equation}
\begin{split}
&\inf \left\{ \lambda \ge 1 \;\middle|\; \forall i,\,\, \lambda F_i - M_i \in \text{cone}^*(\mathbb{F}),\, \{F_i\} \in \mathbb{M}_{\mathbb{F}} \right\} \\
&\qquad = 1 + \inf \left\{ r \ge 0 \;\middle|\; \frac{\mathcal{M} + r\mathcal{W}}{1+r} \in \mathbb{M}_{\mathbb{F}},\,\,\mathcal{W} = \{W_i\},\ W_i \in \text{cone}^*(\mathbb{F})  \right\}.
\end{split}
\label{eq:2}
\end{equation}
\vspace{0.5em}
Since the left-hand side equals $\max_{A_{\mathbb{F}}} P_{\text{succ}}(A_{\mathbb{F}}, \mathcal{M}) / \max_{\mathcal{F} \in \mathbb{M}_{\mathbb{F}}} P_{\text{succ}}(A_{\mathbb{F}}, \mathcal{F})$, combining with \eqref{eq:2} gives
\begin{equation}
\max_{A_{\mathbb{F}}} \frac{P_{\text{succ}}(A_{\mathbb{F}}, \mathcal{M})}{\max_{\mathcal{F} \in \mathbb{M}_{\mathbb{F}}} P_{\text{succ}}(A_{\mathbb{F}}, \mathcal{F})} = 1 + R_{\mathbb{M}_{\mathbb{F}}}^{\mathbb{F}}(\mathcal{M}).
\end{equation}
This concludes the proof.
\end{proof}

\section{Properties of the witness robustness}
\label{sec:properties}
This section collects the basic properties of the witness robustness. The first couples properties are a set of inequalities connecting the sets of free-states and free measurements with witness robustness.
\begin{lemma}[Hierarchy in the free-state set]
\label{lem:mono1}
If $\F_1\subseteq\F_2$, then $\cone^*(\F_1)\supseteq\cone^*(\F_2)$, and consequently
  \begin{equation}
      R_{\Mf}^{\F_1}(\mathcal{M})\le R_{\Mf}^{\F_2}(\mathcal{M}).
  \end{equation}
\end{lemma}

\begin{lemma}[Hierarchy in the free-measurement set]
\label{lem:mono2}
    If $\mathbb{M}_{\F_1}\supseteq\mathbb{M}_{\F_2}$, then
    \begin{equation}
        R_{\mathbb{M}_{\F_1}}^{\F}(\mathcal{M})\le R_{\mathbb{M}_{\F_2}}^{\F}(\mathcal{M}).
    \end{equation}
\end{lemma}
\noindent The proof for both lemmas follows from the fact that a larger feasible set cannot increase the infimum.

Second, as a resource measure for measurements, it must be monotone under free transformations of the measurement.
\begin{lemma}[Monotonicity under free transformations]
\label{lem:monosim}
Let $\Phi$ be a linear map on $N$-tuples of elements of $V^*$ such that (i)~$\Phi(\Mf)\subseteq\Mf$ and (ii)~$\Phi(\mathcal{W})_i\in\cone^*(\F)$ for all $i$ whenever $W_j\in\cone^*(\F)$ for all $j$. Then, for every measurement $\mathcal{M}$ with $\Phi(\mathcal{M})\in\mathbb{M}$,
\begin{equation}
R_{\Mf}^{\F}\big(\Phi(\mathcal{M})\big)\le R_{\Mf}^{\F}(\mathcal{M}).
\end{equation}
\end{lemma}

\begin{proof}
Let $(r,\mathcal{W})$ be feasible for $\mathcal{M}$ in Definition~\ref{def:wr}, i.e.\ $W_i\in\cone^*(\F)$ and $(\mathcal{M}+r\mathcal{W})/(1+r)\in\Mf$. By linearity, $\big(\Phi(\mathcal{M})+r\,\Phi(\mathcal{W})\big)/(1+r) =\Phi\big((\mathcal{M}+r\mathcal{W})/(1+r)\big)\in\Mf$ by~(i), and $\Phi(\mathcal{W})$ is again a witness tuple by~(ii). Hence every $r$ feasible for $\mathcal{M}$ is feasible for $\Phi(\mathcal{M})$, and taking the infimum proves the claim.
\end{proof}

\begin{corollary}[Post- and pre-processing]
\label{cor:procs}
Let $\mathcal{M}=\{M_i\}_{i=1}^N$ be a measurement. (1)~If $q(j|i)\ge0$, $\sum_jq(j|i)=1$, and $\{\sum_iq(j|i)F_i\}_j\in\Mf$ for every $\mathcal{F}=\{F_i\}_i\in\Mf$, then, for $\mathcal{M}'=\{\sum_iq(j|i)M_i\}_j$, we have $R_{\Mf}^{\F}(\mathcal{M}')\le R_{\Mf}^{\F}(\mathcal{M})$. (2)~If $\Lambda\in\mathbb{T}(V,V)$ is a channel with $\Lambda(\F)\subseteq\F$ and $\Lambda^*(\Mf)\subseteq\Mf$, where $\Lambda^*$ denotes the adjoint, then $R_{\Mf}^{\F}\bigl(\{\Lambda^*(M_i)\}_i\bigr)\le R_{\Mf}^{\F}(\mathcal{M})$.
\end{corollary}

\begin{proof}
Both maps are linear, satisfy~(i) by assumption, and return valid measurements, since $\sum_jM'_j=\sum_iM_i=U$ and $\sum_i\Lambda^*(M_i)=\Lambda^*(U)=U$. For~(ii), in case~(1), $\sum_iq(j|i)W_i\in\cone^*(\F)$ because the dual cone is closed under nonnegative combinations; in case~(2), for every $\sigma\in\F$, $\langle\Lambda^*(W_i),\sigma\rangle=\langle W_i,\Lambda(\sigma)\rangle\ge0$ as $\Lambda(\sigma)\in\F$, so $\Lambda^*(W_i)\in\cone^*(\F)$. Both parts follow from Lemma~\ref{lem:monosim}.
\end{proof}

Note that, maps satisfying (i)--(ii) are closed under composition and, by convexity of $\Mf$ and $\cone^*(\F)$, under convex mixtures; Lemma~\ref{lem:monosim} therefore covers arbitrary free simulations---probabilistic combinations of free pre-processings followed by classical post-processings~\cite{Guerini2017,OszmaniecBiswas2019}. 

Similar to generalized and standard robustness, witness robustness also has convexity.
\begin{lemma}[Convexity]
\label{lem:convex}
$R_{\Mf}^{\F}(\mathcal{M})$ is convex in $\mathcal{M}$.
\end{lemma}

\begin{proof}
Using the operational form of Theorem~\ref{thm:main},
\begin{equation}
1+R_{\Mf}^{\F}(\mathcal{M})=\max_{\omega\in\Omega_{\F}}
\frac{\langle\mathcal{M},\omega\rangle}{\max_{\mathcal{F}\in\Mf}\langle\mathcal{F},\omega\rangle}.
\end{equation}
Let $\omega^*$ be the optimal witness for the convex combination $p_1\mathcal{M}_1+p_2\mathcal{M}_2$ with $p_1+p_2=1$. Substituting $\omega^*$ directly removes the outer maximisation and bounds each term separately:
\begin{align}
1+R_{\Mf}^{\F}(p_1\mathcal{M}_1+p_2\mathcal{M}_2)
&= \frac{p_1\langle\mathcal{M}_1,\omega^*\rangle}{\max_{\mathcal{F}\in\Mf}\langle\mathcal{F},\omega^*\rangle}
 + \frac{p_2\langle\mathcal{M}_2,\omega^*\rangle}{\max_{\mathcal{F}\in\Mf}\langle\mathcal{F},\omega^*\rangle}\nonumber\\
&\le p_1\big(1+R_{\Mf}^{\F}(\mathcal{M}_1)\big)+p_2\big(1+R_{\Mf}^{\F}(\mathcal{M}_2)\big)\nonumber\\
&= p_1+p_2 +p_1 R_{\Mf}^{\F}(\mathcal{M}_1)+p_2 R_{\Mf}^{\F}(\mathcal{M}_2).
\end{align}
Since $p_1+p_2=1$, this yields the desired inequality
\begin{equation}
R_{\Mf}^{\F}(p_1\mathcal{M}_1+p_2\mathcal{M}_2)\le
p_1 R_{\Mf}^{\F}(\mathcal{M}_1)+p_2 R_{\Mf}^{\F}(\mathcal{M}_2).
\end{equation}
\end{proof}

\begin{lemma}[Universal upper bound]
\label{lem:bound}
For any $N$-outcome measurement $\mathcal{M}=\{M_i\}_{i=1}^N$, since the trivial measurement $F_i=U/N$ is free ($\mathcal{F}\in\Mf$), then
\begin{equation}
R_{\Mf}^{\F}(\mathcal{M})\le N-1.
\end{equation}
\end{lemma}

\begin{proof}
Choose the trivial measurement $F_i=\tfrac1NU$ and set $r=N-1$. The witness condition requires $(1+r)F_i-M_i\in\cone^*(\F)$. Evaluating directly,
\begin{equation}
(1+r)F_i-M_i = N\big(\tfrac1NU\big)-M_i = U-M_i.
\end{equation}
Since $\mathcal{M}$ is a valid POVM, $U-M_i\ge 0$. Because every free state $\sigma\ge 0$, the positive semidefinite cone is contained in the dual cone, $\mathrm{PSD}\subseteq\cone^*(\F)$; hence $U-M_i\in\cone^*(\F)$ holds for all $i$.
\end{proof}

\begin{remark}[Saturation and quantum data hiding]
When does a measurement reach the bound $R_{\Mf}^{\F}(\mathcal{M})=N-1$? By Theorem~\ref{thm:main} the ratio $P_{\text{succ}}(\mathcal{M})/P_{\text{succ}}(\mathcal{F})$ must reach $N$. As $P_{\text{succ}}(\mathcal{M})\le 1$ and $P_{\text{succ}}(\mathcal{F})\ge 1/N$, saturation occurs if and only if there is a free-state ensemble $A_{\F}^*$ such that
\begin{enumerate}
\item $P_{\text{succ}}(A_{\F}^*,\mathcal{M})=1$ (perfect discrimination by $\mathcal{M}$);
\item $\max_{\mathcal{F}\in\Mf}P_{\text{succ}}(A_{\F}^*,\mathcal{F})=1/N$ (free measurements do no better than random guessing).
\end{enumerate}
This gives a strict operational link between maximal witness robustness and quantum data hiding~\cite{DiVincenzo2002}: achieving the bound requires orthogonal free states whose classical label is completely inaccessible to any free measurement.
\end{remark}

\begin{lemma}[Necessary and sufficient condition for $R^{\F}_{\Mf}(\mathcal{M})=0$]
\label{lem:nc0}
\begin{equation}
R^{\F}_{\Mf}(\mathcal{M})=0 \iff \exists\,\mathcal{F}\in\Mf:\ \forall\sigma\in\F,\ \langle F_i,\sigma\rangle=\langle M_i,\sigma\rangle.
\end{equation}
That is, the witness robustness of $\mathcal{M}$ vanishes if and only if there exists a free measurement $\mathcal{F}$ such that, on every free state, the two measurements produce the same outcome distribution.
\end{lemma}

\begin{proof}
$R^{\F}_{\Mf}(\mathcal{M})=0 \iff \exists\,\mathcal{F}\in\Mf$ with $F_i-M_i\in\cone^*(\F)$, which gives
\begin{equation}
\forall\sigma\in\F:\ \langle F_i,\sigma\rangle=\langle M_i,\sigma\rangle.
\end{equation}
The converse is immediate.
\end{proof}

\begin{remark}[Unfaithful in general]
\label{cor:unfaith}
This condition is in general strictly weaker than the existence of a free measurement $\mathcal{F}$ with $\mathcal{F}=\mathcal{M}$ (non-faithful in general). To recover faithfulness, an additional assumption is needed.
\end{remark}

\section{Faithfulness and the interior condition}
\label{sec:faithful}

While the witness robustness is generically unfaithful (Remark~\ref{cor:unfaith}), faithfulness is recovered if the free-state set $\F$ possesses an interior.

\begin{theorem}[Recovery of faithfulness]
\label{thm:faithful}
If $\F$ has non-empty interior relative to $\Omega(V)$, then $R_{\Mf}^{\F}(\mathcal{M})$ is faithful.
\end{theorem}
\begin{proof}
By Lemma~\ref{lem:nc0}, $R_{\Mf}^{\F}(\mathcal{M})=0$ is equivalent to the existence of a free measurement $\mathcal{F}\in\Mf$ with
\begin{equation}
F_i-M_i\in\cone^*(\F),\qquad\forall i.
\end{equation}
Next we use the geometry of the free state sets. Since $\Omega(V)$ spans $V$, a set with non-empty interior relative to $\Omega(V)$ spans $V$ as well; hence $\cone(\F)$ is full-dimensional and its dual cone, $\cone^*(\F)$, is \emph{pointed}~\cite[Exercise 2.31 (e)]{boyd2004convex}, which by definition~\cite[\S 2.4.1]{boyd2004convex} means
\begin{equation}
\cone^*(\F)\cap-\cone^*(\F)=\{0\},
\end{equation}
where $-\cone^*(\F):=\{-X \mid X\in\cone^*(\F)\}$. Let $X_i=F_i-M_i$. Since both $\mathcal{F}$ and $\mathcal{M}$ are valid measurements, summing over all outcomes gives
\begin{equation}
\sum_i X_i=\sum_i F_i-\sum_i M_i=U-U=0,
\end{equation}
with $U$ the unit effect. For any fixed index $j$ we may therefore isolate
\begin{equation}
-X_j=\sum_{i\neq j}X_i.
\end{equation}
Because the dual cone is closed under addition, $\sum_{i\neq j}X_i\in\cone^*(\F)$, so $-X_j\in\cone^*(\F)$. Together with $X_j\in\cone^*(\F)$ this places $X_j$ in the intersection of the cone and its negative,
\begin{equation}
X_j\in\cone^*(\F)\cap-\cone^*(\F)\implies X_j=0.
\end{equation}
Hence $F_i=M_i$ for all $i$, i.e.\ $\mathcal{M}=\mathcal{F}\in\Mf$. The converse is immediate.
\end{proof}

As Table~\ref{tab:interior} shows, the interior condition holds in the workhorse theories of entanglement, magic, nonlocality, steering, and contextuality, where the witness robustness is therefore faithful; the theories whose free states form a lower-dimensional slice---coherence, asymmetry, and imaginarity---fall outside this guarantee, foreshadowing the next section.

\begin{table}[t]
\centering
\caption{Free-state sets with and without interior across resource theories.}
\label{tab:interior}
\begin{tabular}{@{}llc@{}}
\toprule
Resource theory & Free-state set ($\F$) & Interior? \\
\midrule
Entanglement & Separable / PPT states & Yes \\
Magic (non-stabilizerness) & Stabilizer polytope & Yes \\
Wigner negativity & Wigner-positive states & Yes \\
Bell nonlocality & Local hidden-variable models & Yes \\
Quantum steering & Local hidden-state models & Yes \\
Contextuality & Noncontextual models & Yes \\
\midrule
Coherence & Incoherent (diagonal) states & No \\
Asymmetry & Symmetric states & No \\
Imaginarity & Real states ($\rho=\rho^*$) & No \\
\bottomrule
\end{tabular}
\end{table}

\section{Resource-destroying maps imply zero witness robustness}
\label{sec:rdm}

In resource theories admitting a resource-destroying map~\cite{Liu2017}, the witness robustness vanishes identically for every measurement, and is thus manifestly non-faithful.

\begin{theorem}[Vanishing under resource-destroying maps]
\label{thm:rdm}
Suppose there exists a linear map $\Lambda:V\to V$ with the properties
\begin{enumerate}
\item state-preserving: $\Lambda\big(\Omega(V)\big)\subseteq\Omega(V)$ (for quantum systems: $\Lambda$ is positive and trace-preserving);
\item resource-destroying: $\Lambda(\rho)\in\F$ for all $\rho\in\Omega(V)$;
\item freezing free states: $\Lambda(\sigma)=\sigma$ for all $\sigma\in\F$.
\end{enumerate}
If, in addition, its adjoint $\Lambda^{*}$ is also a resource destroying map for any measurement,
\begin{equation}
\forall\,\mathcal{M}\in\mathbb{M}:\quad \{\Lambda^{*}(M_i)\}_i\in\Mf,
\end{equation}
then the witness robustness of every measurement vanishes,
$R_{\Mf}^{\F}(\mathcal{M})=0$ for all $\mathcal{M}\in\mathbb{M}$.
\end{theorem}

\begin{proof}
For a given $\mathcal{M}$, pick the free measurement constructed by the dual map,
\begin{equation}
F_i=\Lambda^*(M_i),
\end{equation}
which lies in $\Mf$ by assumption. For any free state $\sigma\in\F$,
\begin{align}
\langle F_i-M_i,\sigma\rangle
&= \langle\Lambda^*(M_i)-M_i,\sigma\rangle\nonumber\\
&= \langle M_i,\Lambda(\sigma)\rangle-\langle M_i,\sigma\rangle.
\end{align}
By property (3), $\Lambda(\sigma)=\sigma$ for all $\sigma\in\F$, so
\begin{equation}
\langle M_i,\sigma\rangle-\langle M_i,\sigma\rangle=0.
\end{equation}
Since this holds for every $\sigma\in\F$, it follows that for all $i$
\begin{equation}
F_i-M_i\in\cone^*(\F).
\end{equation}
By using lemma \ref{lem:nc0}, $R_{\Mf}^{\F}(\mathcal{M})=0$ for all $\mathcal{M}$.
\end{proof}

\begin{table}[t]
\centering
\caption{Resource-destroying maps and the induced free measurements.}
\label{tab:rdm}
\small
\setlength{\tabcolsep}{4pt}
\begin{tabular}{@{}lll@{}}
\toprule
Resource theory & Destroying map $\Lambda$ & Induced free measurement \\
\midrule
Coherence & Full dephasing $\Delta(\rho)$ & Diagonal POVMs \\
Asymmetry & Group twirling $\mathcal{G}(\rho)$ & Covariant \\
Imaginarity & $\rho\mapsto\tfrac12(\rho+\rho^*)$ & Real POVMs \\
\bottomrule
\end{tabular}
\end{table}

These hypotheses hold in many standard resource theories. Table~\ref{tab:rdm} lists representative resource-destroying maps $\Lambda$ together with the free measurements induced by their duals $\Lambda^\ast$; in each such theory the witness robustness vanishes for every measurement, $R_{\mathbb{M}_{\mathbb{F}}}^{\mathbb{F}}(\mathcal{M})=0$.

\section{Analytical examples}
\label{sec:examples}

\subsection{Magic (single qubit)}
\label{sec:magic}

To build intuition, we evaluate the witness robustness for a single-qubit toy model with stabilizer states and classically mixed stabilizer
measurements~\cite{Veitch2014,Howard2017,Heinrich2019,BravyiGosset2016}. We write Hermitian operators in the Pauli basis as $A=a_0\mathbb{I}+\vec a\cdot\vec\sigma$, denoted compactly $(a_0,\vec a)$; the trace inner product is
\begin{equation}
\langle(a_0,\vec a),(b_0,\vec b)\rangle
= \mathrm{Tr}\!\left[(a_0\mathbb{I}+\vec a\cdot\vec\sigma)(b_0\mathbb{I}+\vec b\cdot\vec\sigma)\right]
= 2(a_0 b_0+\vec a\cdot\vec b).
\end{equation}

\paragraph{Free states.}
The single-qubit stabilizer states form an octahedron inside the Bloch sphere~\cite{Veitch2014},
\begin{equation}
\mathrm{Stab} := \left\{ \tfrac12(\mathbb{I}+\vec r\cdot\vec\sigma) \;\middle|\;
\|\vec r\|_1=|r_x|+|r_y|+|r_z|\le 1 \right\}.
\end{equation}

\paragraph{Free measurements.}
The free measurements are the POVMs whose effects are conical combinations of stabilizer projectors $\tfrac{\mathbb{I}\pm\sigma_k}{2}$, $\sigma_k$ denotes $k$-th Pauli matrix~\cite{Heimendahl2022},
\begin{equation}
\mathbb{STAB} := \left\{ \{E_i\} \;\middle|\;
E_i=\sum_{k,\pm}c^{(i)}_{k,\pm}\tfrac{\mathbb{I}\pm\sigma_k}{2},\ c^{(i)}_{k,\pm}\ge0,\ \textstyle\sum_i E_i=\mathbb{I} \right\}.
\end{equation}
Equivalently, each effect $E_i=a_0\mathbb{I}+\vec a\cdot\vec\sigma$ satisfies $a_0\ge\|\vec a\|_1$.
This set is precisely the set of single-qubit POVMs implementable by stabilizer operations; see Appendix~\ref{app:stabilizer-measurements} for a proof.

\paragraph{Dual cone.}
An operator $(d_0,\vec d)$ lies in $\cone^*(\mathrm{Stab})$ iff its overlap with every extremal pure stabilizer state is non-negative,
\begin{equation}
\cone^*(\mathrm{Stab})=\left\{(d_0,\vec d)\ \middle|\
\big\langle(d_0,\vec d),(\tfrac12,\pm\tfrac12\hat e_k)\big\rangle\ge 0,\ \forall k=x,y,z\right\}.
\end{equation}
where $\hat{e}_k=(\delta_{jk})_j$ is the $k$-th unit basis vector.
This yields $d_0\pm d_k\ge 0$, which simplifies to the $L_\infty$ condition
\begin{equation}
d_0\ge\|\vec d\|_\infty.
\end{equation}
For a general PSD operator $d_0\mathbb{I}+\vec d\cdot\vec\sigma\ge 0$ the constraint is instead $d_0\ge\|\vec d\|_2$; since $\|\vec d\|_2\ge\|\vec d\|_\infty$, we have $\mathrm{PSD}\subseteq\cone^*(\mathrm{Stab})$, confirming the hierarchical relation $R_{\mathbb{STAB}}^{\mathrm{Stab}}(\mathcal{M})\le R_{\mathbb{STAB}}(\mathcal{M})$.

\paragraph{Robustness of a target measurement.}
Consider a target projective measurement $\mathcal{M}=(\Pi,\mathbb{I}-\Pi)$ with $\Pi=(\tfrac12,\tfrac12\hat n)$, $\|\hat n\|_2=1$, and let the optimal free measurement be $\mathcal{F}=(F_1,F_2)$ parameterised symmetrically as $F_1=(f_0,\vec f)$, $F_2=(1-f_0,-\vec f)$. Membership $\mathcal{F}\in\mathbb{STAB}$ requires
\begin{equation}
f_0\ge\|\vec f\|_1\quad\text{and}\quad 1-f_0\ge\|\vec f\|_1
\implies \min\{f_0,1-f_0\}\ge\|\vec f\|_1.
\end{equation}
The witness robustness condition $(1+r)F_i-M_i\in\cone^*(\mathrm{Stab})$ reads, for $F_1$,
\begin{equation}
(1+r)f_0-\tfrac12\ge\Big\|(1+r)\vec f-\tfrac12\hat n\Big\|_\infty,
\end{equation}
and, for $F_2$,
\begin{equation}
(1+r)(1-f_0)-\tfrac12\ge\Big\|-(1+r)\vec f+\tfrac12\hat n\Big\|_\infty
=\Big\|(1+r)\vec f-\tfrac12\hat n\Big\|_\infty.
\end{equation}
Combining the two,
\begin{equation}
\min\{f_0,1-f_0\}\ge
\frac{\big\|(1+r)\vec f-\tfrac12\hat n\big\|_\infty+\tfrac12}{1+r}.
\end{equation}
Without loss of generality set $f_0=\tfrac12$; the problem reduces to finding the minimal $r$ with
\begin{equation}
\|\vec f\|_1\le\tfrac12\quad\text{and}\quad
\tfrac r2\ge\Big\|(1+r)\vec f-\tfrac12\hat n\Big\|_\infty.
\end{equation}
\paragraph{Geometric optimisation.}
Writing $\vec u:=(1+r)\vec f$, a given $r\ge0$ is feasible iff some $\vec u$ satisfies both
\begin{equation}
\|\vec u\|_1\le\frac{1+r}{2}\qquad\text{and}\qquad
\Big\|\vec u-\tfrac12\hat n\Big\|_\infty\le\frac r2,
\end{equation}
i.e. the $L_\infty$ cube of radius $\tfrac r2$ centred at $\tfrac12\hat n$ intersects the $L_1$ octahedron of radius $\tfrac{1+r}{2}$ centred at $0$. The intersection is non-empty iff the point of the cube closest to the origin in $L_1$ norm lies in the octahedron. This point is found componentwise: $u_k=0$ if $|n_k|\le r$, and $u_k=\tfrac12 n_k-\operatorname{sgn}(n_k)\tfrac r2$ otherwise, giving
\begin{equation}
\min_{\|\vec u-\frac12\hat n\|_\infty\le\frac r2}\|\vec u\|_1
=\tfrac12\sum_k\big(|n_k|-r\big)_+,\qquad (x)_+:=\max(x,0).
\end{equation}
Feasibility of $r$ is therefore the single scalar condition
\begin{equation}
\sum_k\big(|n_k|-r\big)_+\le 1+r,
\end{equation}
and $R^{\mathrm{Stab}}_{\mathbb{STAB}}(\mathcal{M})$ is the least such $r$.

The remaining scalar minimisation can be solved explicitly. Let
\begin{equation}
a_1\ge a_2\ge a_3\ge0
\end{equation}
denote the absolute values of the Bloch-vector components $n_k$ arranged in non-increasing order.

For any real numbers $x_k$, the sum of their positive parts satisfies
\begin{equation}
\sum_k(x_k)_+ = \max_{A\subseteq\{1,2,3\}}\sum_{k\in A}x_k,
\end{equation}
where the maximum is attained by selecting the positive components. Applying this identity to $x_k=a_k-r$ yields
\begin{equation}
\sum_{k=1}^3(a_k-r)_+=\max_{A\subseteq\{1,2,3\}} \left(\sum_{k\in A}a_k-|A|r\right).
\end{equation}
The feasibility condition is consequently equivalent to the collection of linear inequalities
\begin{equation}
\sum_{k\in A}a_k-|A|r\le1+r \qquad
\text{for every }A\subseteq\{1,2,3\},
\end{equation}
or
\begin{equation}
r\ge \frac{\sum_{k\in A}a_k-1}{1+|A|}
\qquad
\text{for every nonempty }A\subseteq\{1,2,3\}.
\end{equation}

For a fixed cardinality $|A|=m$, the strongest lower bound is obtained by choosing the $m$ largest components. The three possible nonempty cardinalities therefore give
\begin{equation}
r\ge\frac{a_1-1}{2},\qquad
r\ge\frac{a_1+a_2-1}{3},\qquad
r\ge\frac{a_1+a_2+a_3-1}{4}.
\end{equation}
Together with $r\ge0$, this gives the explicit piecewise-analytical expression
\begin{equation}
\boxed{
R^{\mathrm{Stab}}_{\mathbb{STAB}}(\mathcal M)
=
\max\left\{
0,\,
\frac{a_1-1}{2},\,
\frac{a_1+a_2-1}{3},\,
\frac{a_1+a_2+a_3-1}{4}
\right\}.
}
\label{eq:magic-explicit}
\end{equation}
Hence the original geometric optimisation reduces completely to sorting the three absolute Bloch components and evaluating a finite maximum. Figure~\ref{fig:magic} illustrates this expression along two Bloch-vector trajectories and confirms its agreement with direct numerical evaluation.

\begin{remark}[Generalized robustness]
For the generalized robustness $R_{\mathbb{STAB}}(\mathcal{M})$ the noise $\mathcal{W}$ must be a valid measurement, imposing the PSD condition $d_0\ge\|\vec d\|_2$ in place of the dual cone's $L_\infty$ constraint. Geometrically the cube becomes an $L_2$ ball, and the optimisation seeks the minimal $r$ with $\|\vec u\|_1\le\tfrac{1+r}{2}$ and $\|\vec u-\tfrac12\hat n\|_2\le\tfrac r2$. The two measures therefore probe different geometries, and (Fig.~\ref{fig:magic}) their values can even move in opposite directions along a Bloch trajectory.
\end{remark}

\begin{figure*}[t]
  \centering
  \includegraphics[width=\textwidth]{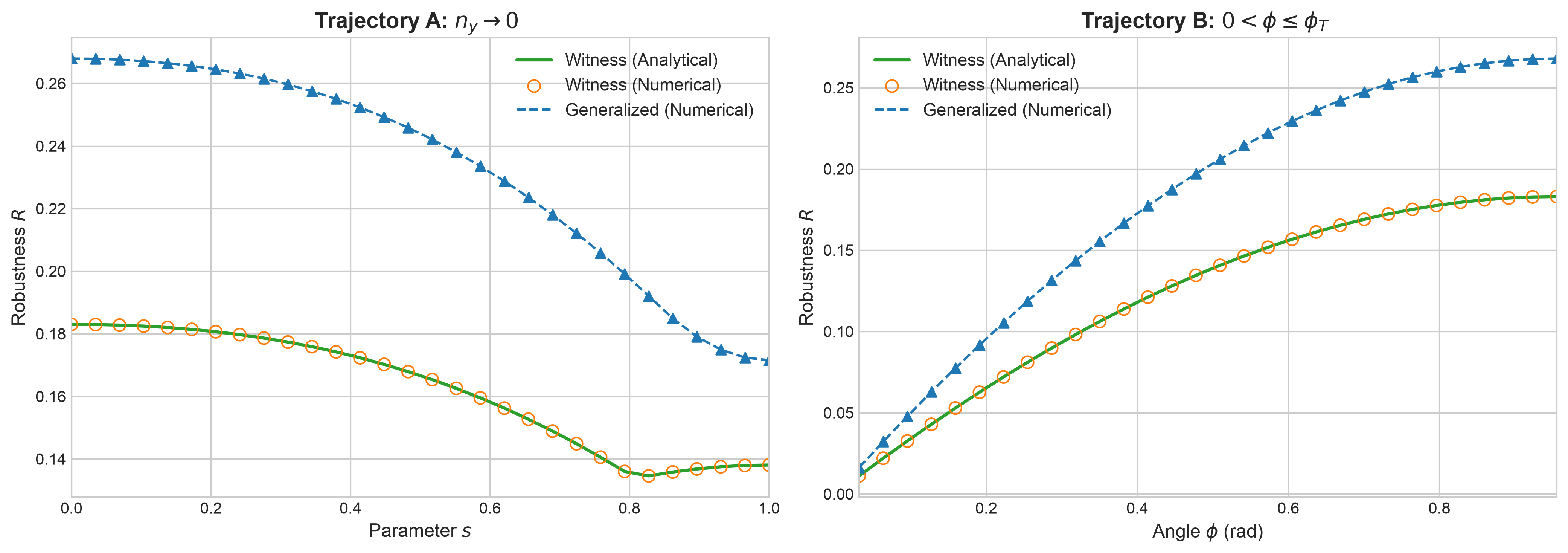}
  \caption{\textbf{Witness robustness and generalized robustness along two single-qubit Bloch-vector trajectories.} For the projective measurement $\mathcal{M}=(\Pi,I-\Pi)$, with $\Pi=\tfrac12(I+\hat n\cdot\vec\sigma)$ and $\|\hat n\|_2=1$, the witness robustness is computed analytically (green line) and numerically (orange circles), while the generalized robustness is computed numerically (blue dashed line and triangles). \textbf{(a)} Along $\hat n(s)\propto(1,1-s,1)$, $s\in[0,1]$, the witness robustness develops a kink near $s=0.81$ and subsequently increases, whereas the generalized robustness continues to decrease. \textbf{(b)} The trajectory $\hat n(\phi)=(\sin\phi/\sqrt2,\sin\phi/\sqrt2,\cos\phi)$ runs from $|0\rangle$ to the $T$ state, with $\phi\in[0,\arccos(1/\sqrt3)]$.}
  \label{fig:magic}
\end{figure*}

\subsection{PPT entanglement (two qubits)}
\label{sec:ppt}

In the theory of bipartite entanglement, a state $\sigma$ is PPT if its partial transpose satisfies $\sigma^\Gamma\ge 0$~\cite{Peres1996,Horodecki1996,HorodeckiReview2009}. As the free measurements we take the \emph{PPT measurements}---POVMs all of whose effects have a positive partial transpose,
\begin{equation}
\mathbb{PPT} := \Big\{ \{F_i\}_i \;\Big|\;
F_i\ge 0,\ F_i^\Gamma\ge 0\ \ \forall i,\ \ \textstyle\sum_i F_i=\mathbb{I} \Big\},
\end{equation}
the standard free class in quantum data hiding and PPT state discrimination~\cite{DiVincenzo2002}. It contains all LOCC and separable measurements and, unlike the separable class, is SDP-representable.

The PPT dual cone has been well studied~\cite{HorodeckiReview2009}: an operator $W$ is a decomposable entanglement witness if $W=P+Q^\Gamma$ for some $P,Q\ge 0$, and the dual cone of PPT state cone is exactly the set of such witnesses,
\begin{equation}
\cone^*(\mathrm{PPT})=\{ W\mid W=P+Q^\Gamma,\ P,Q\ge 0 \}.
\end{equation}
We compute $R_{\mathbb{PPT}}^{\mathrm{PPT}}(\mathcal{M})$ for the projective measurement $\mathcal{M}=\{\Pi,\mathbb{I}-\Pi\}$ onto the entangled pure state $|\psi\rangle=\lambda_1|00\rangle+\lambda_2|11\rangle$.

\paragraph{Decomposition of the partial transpose.}
The partial transpose of $\Pi=|\psi\rangle\langle\psi|$ over the second subsystem is
\begin{equation}
\Pi^\Gamma=\lambda_1^2|00\rangle\langle00|+\lambda_2^2|11\rangle\langle11|+\lambda_1\lambda_2(|01\rangle\langle10|+|10\rangle\langle01|),
\end{equation}
which we write as a difference $\Pi^\Gamma=P-Q$ of positive operators,
\begin{align}
P &:= \lambda_1^2|00\rangle\langle00|+\lambda_2^2|11\rangle\langle11|
     +\lambda_1\lambda_2|\Psi^+\rangle\langle\Psi^+|\ge 0,\\
Q &:= \lambda_1\lambda_2|\Psi^-\rangle\langle\Psi^-|\ge 0,
\end{align}
with $|\Psi^\pm\rangle=\tfrac1{\sqrt2}(|01\rangle\pm|10\rangle)$. Since
$|00\rangle,|11\rangle,|\Psi^+\rangle$ are mutually orthogonal, $P$ is diagonal in that basis with spectrum $\{\lambda_1^2,\lambda_2^2,\lambda_1\lambda_2\}\subset[0,1]$.

\paragraph{Construction of a free measurement $\mathcal{F}\in\mathbb{PPT}$.}
Parameterise the first effect through its partial transpose, $(1+r)F_1^\Gamma:=P$, set $F_2=\mathbb{I}-F_1$, and choose $r=\tfrac12\lambda_1\lambda_2$. To verify that $\mathcal{F}$ is a valid PPT measurement we check both effects and their partial transposes:
\begin{itemize}
\item \emph{Positivity of partial transposes.} By construction
      $F_1^\Gamma=P/(1+r)\ge 0$, and its largest eigenvalue is
      $\max(\lambda_1^2,\lambda_2^2,\lambda_1\lambda_2)/(1+r)\le 1$, so
      $F_2^\Gamma=\mathbb{I}-F_1^\Gamma\ge 0$.
\item \emph{Positivity of effects.} From the parameterisation
      $(1+r)F_1=P^\Gamma=\Pi+Q^\Gamma$, which in the basis
      $\{|00\rangle,|11\rangle,|01\rangle,|10\rangle\}$ is block-diagonal,
\begin{equation}
P^\Gamma=\begin{pmatrix}\lambda_1^2 & \tfrac12\lambda_1\lambda_2\\
\tfrac12\lambda_1\lambda_2 & \lambda_2^2\end{pmatrix}
\oplus
\begin{pmatrix}\tfrac12\lambda_1\lambda_2 & 0\\ 0 & \tfrac12\lambda_1\lambda_2\end{pmatrix}.
\end{equation}
      The first block has trace $\lambda_1^2+\lambda_2^2=1$ and determinant
      $\tfrac34\lambda_1^2\lambda_2^2\ge0$, so its eigenvalues are non-negative
      and sum to $1$; the second block is diagonal with entries
      $\tfrac12\lambda_1\lambda_2\in[0,1]$. Hence the spectrum of $P^\Gamma$ lies
      in $[0,1]$, giving $F_1=P^\Gamma/(1+r)\ge 0$ and $F_2=\mathbb{I}-F_1\ge 0$.
\end{itemize}
Thus $\mathcal{F}$ is a valid PPT measurement.

\paragraph{Verifying the witness robustness condition.}
We require $(1+r)F_i-M_i\in\cone^*(\mathrm{PPT})$:
\begin{itemize}
\item for $F_1$: $(1+r)F_1-\Pi=P^\Gamma-\Pi=Q^\Gamma$, and since $Q\ge0$ this is a valid PPT witness;
\item for $F_2$: $(1+r)F_2-(\mathbb{I}-\Pi)=r\mathbb{I}-Q^\Gamma$, which requires $r\ge\max_{\sigma\in\mathrm{PPT}}\langle Q^\Gamma,\sigma\rangle$.
\end{itemize}
Using $\langle Q^\Gamma,\sigma\rangle=\mathrm{Tr}(Q\sigma^\Gamma) =\lambda_1\lambda_2\langle\Psi^-|\sigma^\Gamma|\Psi^-\rangle$ and the identity $|\Psi^-\rangle\langle\Psi^-|=\tfrac12(\mathbb{I}-\mathrm{SWAP})$,
\begin{align}
\langle\Psi^-|\sigma^\Gamma|\Psi^-\rangle
&= \tfrac12\big(1-\mathrm{Tr}(\mathrm{SWAP}\cdot\sigma^\Gamma)\big)\nonumber\\
&= \tfrac12\big(1-\mathrm{Tr}(\mathrm{SWAP}^\Gamma\cdot\sigma)\big)\nonumber\\
&= \tfrac12\big(1-\langle\Omega|\sigma|\Omega\rangle\big),
\end{align}
with $|\Omega\rangle=\sum_i|ii\rangle$. As $\sigma\in\mathrm{PPT}\implies\sigma\ge0$, we have $\langle\Omega|\sigma|\Omega\rangle\ge0$, so the supremum is bounded by $\tfrac12\lambda_1\lambda_2$. Our choice $r=\tfrac12\lambda_1\lambda_2$ saturates it, giving the upper bound
\begin{equation}
R_{\mathbb{PPT}}^{\mathrm{PPT}}(\mathcal{M})\le\tfrac12\lambda_1\lambda_2
\end{equation}

\paragraph{Lower bound via operational advantage.}
By using Theorem~\ref{thm:main},
\begin{equation}
1+R^{\mathrm{PPT}}_{\mathbb{PPT}}(\mathcal M)=\max_{A_{\mathrm{PPT}}}
\frac{P_{\text{succ}}(A_{\mathrm{PPT}},\mathcal M)}{\max_{\mathcal F\in\mathbb{PPT}}P_{\text{succ}}(A_{\mathrm{PPT}},\mathcal F)},
\quad P_{\text{succ}}(A,\mathcal F)=\sum_i P_i\langle F_i,\sigma_i\rangle,
\end{equation}
the maximisation being over free-state ensembles
$A_{\mathrm{PPT}}=\{(P_i,\sigma_i)\}_{i=1,2}$ with $\sigma_i\in\mathrm{PPT}$. Any single ensemble thus lower-bounds $1+R^{\mathrm{PPT}}_{\mathbb{PPT}}$. With $|\Phi^+\rangle=\tfrac1{\sqrt2}(|00\rangle+|11\rangle)$, choose the weighted ensemble
\begin{equation}
P_1=\tfrac13,\ \sigma_1=\tfrac16\mathbb I+\tfrac13|\Phi^+\rangle\langle\Phi^+|;
\qquad
P_2=\tfrac23,\ \sigma_2=\tfrac13\big(\mathbb I-|\Phi^+\rangle\langle\Phi^+|\big).
\end{equation}

\emph{(1) Both states are PPT.} Using
$|\Phi^+\rangle\langle\Phi^+|^{\Gamma}=\tfrac12\mathrm{SWAP}$,
\begin{equation}
\sigma_1^{\Gamma}=\tfrac16(\mathbb I+\mathrm{SWAP})\ge 0,
\qquad
\sigma_2^{\Gamma}=\tfrac13\mathbb I-\tfrac16\mathrm{SWAP}\ge 0,
\end{equation}
since the eigenvalues of $\mathrm{SWAP}$ are $\pm1$; primal positivity
$\sigma_1,\sigma_2\ge0$ is immediate.

\emph{(2) Numerator.} With
$|\langle\psi|\Phi^+\rangle|^2=\tfrac12(\lambda_1+\lambda_2)^2=\tfrac12(1+2\lambda_1\lambda_2)$,
\begin{equation}
\langle\Pi,\sigma_1\rangle=\tfrac16+\tfrac13|\langle\psi|\Phi^+\rangle|^2=\tfrac{1+\lambda_1\lambda_2}{3},
\quad
\langle\mathbb I-\Pi,\sigma_2\rangle=\tfrac{5+2\lambda_1\lambda_2}{6},
\end{equation}
so the target success probability is
\begin{equation}
P_{\text{succ}}(A_{\mathrm{PPT}},\mathcal M) =\tfrac13\cdot\tfrac{1+\lambda_1\lambda_2}{3}+\tfrac23\cdot\tfrac{5+2\lambda_1\lambda_2}{6} =\frac{2+\lambda_1\lambda_2}{3}.
\end{equation}

\emph{(3) Denominator.} For any PPT measurement $\mathcal F=(F_1,\mathbb I-F_1)$,
\begin{equation}
P_{\text{succ}}(A_{\mathrm{PPT}},\mathcal F)=P_2+\langle F_1,\Delta\rangle, \quad \Delta:=P_1\sigma_1-P_2\sigma_2=-\tfrac16\mathbb I+\tfrac13|\Phi^+\rangle\langle\Phi^+|.
\end{equation}
Crucially $-\Delta$ is a decomposable witness: using
$|\Psi^-\rangle\langle\Psi^-|^{\Gamma}=\tfrac12\mathbb I-|\Phi^+\rangle\langle\Phi^+|$,
\begin{equation}
-\Delta=\tfrac16\mathbb I-\tfrac13|\Phi^+\rangle\langle\Phi^+|
=\Big(\tfrac13|\Psi^-\rangle\langle\Psi^-|\Big)^{\Gamma}=:Y^{\Gamma},
\quad Y\succeq0.
\end{equation}
Hence for every PPT effect $F_1$ (which has $F_1^{\Gamma}\succeq0$),
\begin{equation}
\langle F_1,\Delta\rangle=-\langle F_1,Y^{\Gamma}\rangle=-\mathrm{Tr}\!\big(F_1^{\Gamma}Y\big)\le 0,
\end{equation}
with equality at $F_1=0$. The optimal free strategy is therefore the trivial measurement $(0,\mathbb I)$, giving $\max_{\mathcal F\in\mathbb{PPT}}P_{\text{succ}}(A_{\mathrm{PPT}},\mathcal F)=P_2=\tfrac23$.

\emph{(4) Conclusion.} The advantage ratio for this ensemble is
\begin{equation}
\frac{P_{\text{succ}}(A_{\mathrm{PPT}},\mathcal M)}{\max_{\mathcal F\in\mathbb{PPT}}P_{\text{succ}}(A_{\mathrm{PPT}},\mathcal F)}
=\frac{(2+\lambda_1\lambda_2)/3}{2/3}=1+\tfrac12\lambda_1\lambda_2,
\end{equation}
so by Theorem~\ref{thm:main}, $R^{\mathrm{PPT}}_{\mathbb{PPT}}(\mathcal M)\ge\tfrac12\lambda_1\lambda_2$.
Combined with the upper bound,
\begin{equation}
\boxed{\,R^{\mathrm{PPT}}_{\mathbb{PPT}}(\mathcal M)=\tfrac12\lambda_1\lambda_2=\frac{C(|\psi\rangle)}{4}\,},
\end{equation}
where $C(|\psi\rangle)$ is the concurrence~\cite{Wootters1998} of the target pure state. The witness robustness of the projector onto $|\psi\rangle$ is thus, up to a factor of four, exactly the concurrence of the state it projects onto: vanishing for product states and peaking at $C=1$ (so $R=1/4$) for the maximally entangled projector.

\begin{figure}[t]
  \centering
  \includegraphics[width=0.75\columnwidth]{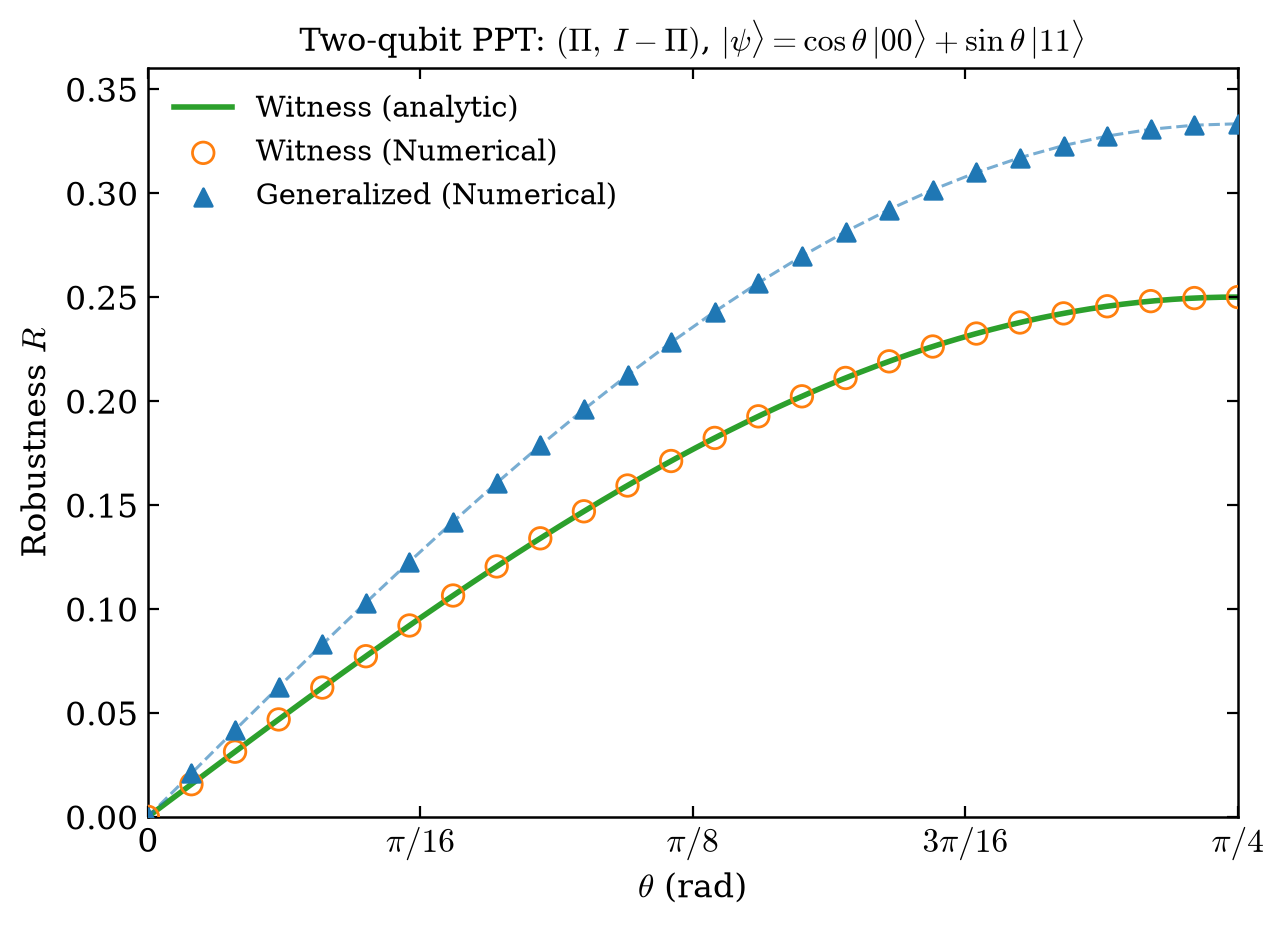}
  \caption{\textbf{Witness robustness and generalized robustness of a two-qubit projective measurement.} The measurement is $\mathcal{M}=(\Pi,\mathbb{I}-\Pi)$, where $\Pi=|\psi(\theta)\rangle\langle\psi(\theta)|$ and $|\psi(\theta)\rangle=\cos\theta|00\rangle+\sin\theta|11\rangle$, with $\theta\in[0,\pi/4]$. The witness robustness $R_{\mathbb{PPT}}^{\mathrm{PPT}}=C/4$, where $C=\sin(2\theta)$ is the concurrence of $|\psi(\theta)\rangle$, is evaluated analytically (green line) and numerically (orange circles). The generalized robustness is evaluated numerically (blue dashed line and triangles).}
  \label{fig:ppt}
\end{figure}

\section{Conclusion}
\label{sec:conclusion}
In this work, we introduced the witness robustness of quantum measurements and established its operational significance in the discrimination of free-state ensembles. A nonzero witness robustness guarantees the existence of a free-state ensemble for which the measurement achieves a strictly higher discrimination success probability than any free measurement. We then investigated several fundamental properties of this quantity, including convexity and monotonicity. Most importantly, unlike the standard and generalized robustness measures, witness robustness is not faithful in general. Its faithfulness requires additional structure in the underlying resource theory. Operationally, this means that not every resourceful measurement necessarily provides an advantage in the discrimination of free-state ensembles.

We identified two settings that illustrate this distinctive behavior. First, when the set of free states has nonempty interior, we proved that witness robustness is faithful. Second, when the resource theory admits a resource-destroying map, we showed that the witness robustness of every measurement vanishes. Consequently, in such resource theories, every ensemble of free states can be optimally discriminated by a free measurement, and resourceful measurements offer no advantage in this task. We further evaluated witness robustness in two concrete examples: the single-qubit resource theory of magic and the PPT-based resource theory of two-qubit entanglement. In both cases, the analytical results were found to agree with independent numerical calculations.

The framework developed here suggests several directions for future research. One natural direction is to extend the same construction to other discrimination tasks. For instance, one may consider the discrimination of free subchannels using free measurements. The corresponding operational advantage would then be quantified by a new robustness measure for states whose admissible noise is no longer restricted to the state space, mirroring the way in which witness robustness allows measurement noise to extend beyond the set of physical measurements. It would also be interesting to investigate more general decision-theoretic settings, including discrimination games in which different guesses are assigned different scores, as well as unambiguous discrimination tasks.

Another important direction is to move beyond the convex setting considered throughout this work. The present construction and several of its properties rely essentially on the convexity of the underlying resource theory. When the free set is nonconvex, the geometry of the problem changes substantially, and properties such as convexity, monotonicity, faithfulness, and the operational interpretation of witness robustness may require new formulations or may fail altogether. Resource theory of quantum discord, whose sets of free states are generally nonconvex, provides a particularly relevant setting in which to explore these questions. More broadly, these directions suggest a family of witness-based robustness measures adapted to different resource-theoretic decision problems, with their behavior determined not only by the operational task but also by the convex or nonconvex geometry of the underlying free set. \\

\textbf{\textit{Acknoledgement}}--- This work was supported by the National Science Centre Poland (Grant No. 2022/46/E/ST2/00115 and 2024/55/B/ST2/01590).

\appendix

\section{Table of notation}
\label{app:notation}

To facilitate reading, the key symbols used throughout are summarised below.
\begin{table}[h]
\centering
\renewcommand{\arraystretch}{1.3}
\begin{tabular}{@{}p{0.18\columnwidth} p{0.74\columnwidth}@{}}
\toprule
Symbol & Description \\
\midrule
$\Omega(V)$ & The state space defined on the real vector space $V$. \\
$\mathbb{M}$ & The complete set of all valid measurements. \\
$\F$ & The closed convex set of free states. \\
$\Mf$ & The closed convex set of free measurements. \\
$\mathcal{M}=\{M_i\}$ & A target measurement and its effects. \\
$\mathcal{F}=\{F_i\}$ & A free measurement and its effects ($\mathcal{F}\in\Mf$). \\
$\cone(\F)$ & The conical hull of the free-state set. \\
$\cone^*(\F)$ & The dual cone of $\F$: the set of all free-state witnesses. \\
$R^{\mathrm{std}}(\mathcal{M})$ & The standard robustness of a measurement. \\
$R(\mathcal{M})$ & The generalized robustness of a measurement. \\
$R^{\F}_{\Mf}(\mathcal{M})$ & The witness robustness of a measurement (this work). \\
$A_{\F}=\{p_i,\sigma_i\}$ & A free-state ensemble, $\sigma_i\in\F$. \\
$P_{\text{succ}}(A_{\F},\mathcal{M})$ & Discrimination success probability for $A_{\F}$ using $\mathcal{M}$. \\
$\Lambda,\Lambda^*$ & A resource-destroying map and its dual on the measurement space. \\
$\Gamma$ & The partial transpose over a bipartite subsystem. \\
$\mathrm{PPT},\mathbb{PPT}$ & PPT states and PPT measurements, respectively. \\
\bottomrule
\end{tabular}
\end{table}
\section{Characterization of single-qubit stabilizer measurements}
\label{app:stabilizer-measurements}

We show that the free-measurement set used in Sec.~\ref{sec:magic} is exactly
the set of POVMs implementable by stabilizer operations.  Stabilizer operations
are understood in their full operational sense, allowing arbitrary multiqubit
stabilizer ancillas, joint Clifford transformations, adaptive Pauli
measurements, discarding, classical randomness, and classical
post-processing~\cite{Heimendahl2022}.  Let
\begin{equation}
P_{k,s}:=\frac{\mathbb I+s\sigma_k}{2},
\qquad k\in\{x,y,z\},\quad s\in\{+1,-1\}.
\end{equation}

\begin{proposition}[Single-qubit stabilizer POVMs]
\label{prop:single-qubit-stabilizer-povms}
For POVMs on a single-qubit input with an arbitrary finite number of outcomes,
\begin{equation}
\mathsf M_{\mathrm{SO}}^{(1)}
=\mathbb{STAB},
\end{equation}
where $\mathsf M_{\mathrm{SO}}^{(1)}$ denotes the single-qubit POVMs whose
classical-output measurement channels are implementable by stabilizer
operations.
\end{proposition}

\begin{proof}
Let $\mathcal M=\{E_i\}_{i=1}^N\in\mathsf M_{\mathrm{SO}}^{(1)}$ and associate
with it the classical-output measurement channel
\begin{equation}
\mathcal Q_{\mathcal M}(\rho)
=\sum_{i=1}^N\operatorname{Tr}(E_i\rho)|i\rangle\langle i|.
\end{equation}
The outcomes are encoded as distinct computational-basis states of a
multiqubit register; unused basis states, if any, are assigned zero
probability.  Prepare the Bell stabilizer state
\begin{equation}
|\Phi^+\rangle_{RA}=\frac{|00\rangle+|11\rangle}{\sqrt2}
\end{equation}
and apply $\mathcal Q_{\mathcal M}$ to $A$.  The identity on $R$ together with
the stabilizer protocol on $A$ is again a stabilizer operation.  Therefore the
normalized Choi state
\begin{equation}
J_{\mathcal Q_{\mathcal M}}
:=(\operatorname{id}_R\otimes\mathcal Q_{\mathcal M})
(|\Phi^+\rangle\langle\Phi^+|)
=\frac12\sum_{i=1}^N E_i^{\mathsf T}\otimes|i\rangle\langle i|
\end{equation}
is a stabilizer state.  Measuring the output register in the computational
basis and postselecting outcome $i$ is also a stabilizer operation.  It leaves
the subnormalized stabilizer state $E_i^{\mathsf T}/2$ on $R$, and hence
\begin{equation}
E_i^{\mathsf T}\in\cone\{P_{k,s}\}_{k,s}.
\end{equation}
Transposition fixes the $X$- and $Z$-eigenprojectors and exchanges the two
$Y$-eigenprojectors,
\begin{equation}
P_{x,s}^{\mathsf T}=P_{x,s},\qquad
P_{z,s}^{\mathsf T}=P_{z,s},\qquad
P_{y,s}^{\mathsf T}=P_{y,-s}.
\end{equation}
Thus $E_i\in\cone\{P_{k,s}\}_{k,s}$ for every $i$, proving
$\mathsf M_{\mathrm{SO}}^{(1)}\subseteq\mathbb{STAB}$.

Conversely, suppose $\{E_i\}_{i=1}^N\in\mathbb{STAB}$.  Write
\begin{equation}
E_i=\sum_{k,s}c_{k,s}^{(i)}P_{k,s},
\qquad c_{k,s}^{(i)}\ge0,
\end{equation}
and define $C_{k,s}:=\sum_i c_{k,s}^{(i)}$.  Comparing the coefficients of
$\mathbb I,\sigma_x,\sigma_y,\sigma_z$ in $\sum_iE_i=\mathbb I$ gives
\begin{equation}
C_{k,+}=C_{k,-}=:q_k,
\qquad \sum_kq_k=1.
\end{equation}
For $q_k>0$, define
\begin{equation}
p(i|k,s):=\frac{c_{k,s}^{(i)}}{q_k};
\end{equation}
when $q_k=0$, choose $p(i|k,s)$ arbitrarily.  Then
\begin{equation}
E_i=\sum_{k,s}q_kp(i|k,s)P_{k,s}.
\end{equation}
The POVM is therefore implemented by choosing $k$ with probability $q_k$,
measuring $\sigma_k$, and classically mapping the outcome $s$ to $i$ according
to $p(i|k,s)$.  This is a stabilizer operation, proving
$\mathbb{STAB}\subseteq\mathsf M_{\mathrm{SO}}^{(1)}$.
\end{proof}
\section{Numerical evaluation of the analytical examples}
\label{app:numerical-examples}

Here we give the four optimization problems used to generate the
numerical markers in Figs.~\ref{fig:magic} and~\ref{fig:ppt}. For each sampled value of the trajectory parameter, the corresponding program was solved and its optimal value of \(r\) was plotted. The green curves in the figures are obtained independently from the analytical expressions.

\subsection{Single-qubit magic}

Consider the binary projective measurement
\begin{equation}
\mathcal{M}=(\Pi,\mathbb{I}-\Pi),
\qquad
\Pi=\left(\frac12,\frac12\hat n\right),
\qquad
\|\hat n\|_2=1,
\end{equation}
where we use the Pauli representation
\(A=(a_0,\vec a)=a_0\mathbb{I}+\vec a\cdot\vec\sigma\).
As shown in Sec.~\ref{sec:magic}, the optimizing free measurement may be
written as
\begin{equation}
\mathcal{F}=(F_1,F_2),
\qquad
F_1=\left(\frac12,\vec f\right),
\qquad
F_2=\left(\frac12,-\vec f\right),
\qquad
\|\vec f\|_1\le\frac12.
\end{equation}
Following the notation of the main text, define
\begin{equation}
\vec u:=(1+r)\vec f.
\end{equation}

The numerical values of the witness robustness shown by the orange circles
in Fig.~\ref{fig:magic} are obtained from
\begin{align}
R_{\mathbb{STAB}}^{\mathrm{Stab}}(\mathcal{M})
={}&
\min_{r,\vec u}\ r
\nonumber\\
\text{subject to}\quad
&r\ge0,
\nonumber\\
&\|\vec u\|_1\le\frac{1+r}{2},
\nonumber\\
&\left\|\vec u-\frac{\hat n}{2}\right\|_\infty
\le\frac r2.
\label{eq:numerical-magic-witness}
\end{align}

For comparison, the numerical values of the generalized robustness shown
by the blue triangles are obtained from
\begin{align}
R_{\mathbb{STAB}}(\mathcal{M})
={}&
\min_{r,\vec u}\ r
\nonumber\\
\text{subject to}\quad
&r\ge0,
\nonumber\\
&\|\vec u\|_1\le\frac{1+r}{2},
\nonumber\\
&\left\|\vec u-\frac{\hat n}{2}\right\|_2
\le\frac r2.
\label{eq:numerical-magic-generalized}
\end{align}

Thus, the two numerical optimizations differ only in the final norm:
the witness robustness uses the stabilizer-state witness condition,
represented by the \(\ell_\infty\) norm, whereas the generalized robustness
uses positivity on all qubit states, represented by the Euclidean norm.

\subsection{Two-qubit PPT entanglement}

Let
\begin{equation}
\mathcal{M}_\theta=\{M_1,M_2\},
\qquad
M_1=|\psi(\theta)\rangle\langle\psi(\theta)|,
\qquad
M_2=\mathbb I-M_1,
\end{equation}
where
\begin{equation}
|\psi(\theta)\rangle
=
\cos\theta\,|00\rangle+\sin\theta\,|11\rangle.
\end{equation}
We again introduce the scaled free effects \(G_i=(1+r)F_i\). Denote partial
transposition on the second qubit by \(\Gamma\). The free-measurement
conditions are
\begin{equation}
G_1+G_2=(1+r)\mathbb I,
\qquad
G_i\succeq0,
\qquad
G_i^\Gamma\succeq0.
\end{equation}

Since the dual cone of the PPT cone is the decomposable cone,
\begin{equation}
\operatorname{cone}^{*}(\mathrm{PPT})
=
\left\{
P+Q^\Gamma\;\middle|\;P,Q\succeq0
\right\},
\end{equation}
the witness robustness, shown by the orange circles in Fig.~\ref{fig:ppt}, is obtained from
\begin{equation}
\begin{aligned}
R_{\mathbb{PPT}}^{\mathrm{PPT}}(\mathcal{M}_\theta)
={}&\min_{r,\{G_i,P_i,Q_i\}}\ r\\
\text{subject to}\quad
&r\geq0,\\
&G_1+G_2=(1+r)\mathbb I,\\
&G_i\succeq0,\qquad G_i^\Gamma\succeq0,\\
&G_i-M_i=P_i+Q_i^\Gamma,\\
&P_i\succeq0,\qquad Q_i\succeq0,
\qquad i=1,2.
\end{aligned}
\label{eq:ppt-witness-numerical}
\end{equation}

For the generalized robustness, the noise effects must be positive
semidefinite. Thus, the decomposable-cone constraint is replaced by
\(G_i-M_i\succeq0\). The blue triangles in Fig.~\ref{fig:ppt} are obtained
from
\begin{equation}
\begin{aligned}
R_{\mathbb{PPT}}(\mathcal{M}_\theta)
={}&\min_{r,\{G_i\}}\ r\\
\text{subject to}\quad
&r\geq0,\\
&G_1+G_2=(1+r)\mathbb I,\\
&G_i\succeq0,\qquad G_i^\Gamma\succeq0,\\
&G_i-M_i\succeq0,
\qquad i=1,2.
\end{aligned}
\label{eq:ppt-generalized-numerical}
\end{equation}
\bibliographystyle{quantum}
\bibliography{refs}

@Article{ChitambarGour2019,
  author        = {Chitambar, Eric and Gour, Gilad},
  journal       = {Reviews of Modern Physics},
  title         = {{Quantum resource theories}},
  year          = {2019},
  pages         = {025001},
  volume        = {91},
  archiveprefix = {arXiv},
  doi           = {10.1103/RevModPhysical91.025001},
  eprint        = {1806.06107},
}

@Article{VidalTarrach1999,
  author  = {Vidal, G. and Tarrach, R.},
  journal = {Physical Review A},
  title   = {{Robustness of entanglement}},
  year    = {1999},
  pages   = {141},
  volume  = {59},
  doi     = {10.1103/PhysRevA.59.141},
}

@Article{TakagiRegula2019,
  author        = {Takagi, Ryuji and Regula, Bartosz},
  journal       = {Physical Review X},
  title         = {{General resource theories in quantum mechanics and beyond: Operational characterization via discrimination tasks}},
  year          = {2019},
  pages         = {031053},
  volume        = {9},
  archiveprefix = {arXiv},
  doi           = {10.1103/PhysRevX.9.031053},
  eprint        = {1901.08127},
}

@Article{OszmaniecBiswas2019,
  author        = {Oszmaniec, Micha{\l} and Biswas, Tanmoy},
  journal       = {Quantum},
  title         = {{Operational relevance of resource theories of quantum measurements}},
  year          = {2019},
  pages         = {133},
  volume        = {3},
  archiveprefix = {arXiv},
  doi           = {10.22331/q-2019-04-26-133},
  eprint        = {1901.08566},
}

@Article{Skrzypczyk2019,
  author  = {Skrzypczyk, Paul and {\v S}upi{\'c}, Ivan and Cavalcanti, Daniel},
  journal = {Physical Review Letters},
  title   = {{All sets of incompatible measurements give an advantage in quantum state discrimination}},
  year    = {2019},
  pages   = {130403},
  volume  = {122},
  doi     = {10.1103/PhysRevLett.122.130403},
}

@Article{Uola2019,
  author  = {Uola, Roope and Kraft, Tristan and Shang, Jiangwei and Yu, Xiao-Dong and G{\"u}hne, Otfried},
  journal = {Physical Review Letters},
  title   = {{Quantifying quantum resources with conic programming}},
  year    = {2019},
  pages   = {130404},
  volume  = {122},
  doi     = {10.1103/PhysRevLett.122.130404},
}

@Article{Liu2017,
  author        = {Liu, Zi-Wen and Hu, Xueyuan and Lloyd, Seth},
  journal       = {Physical Review Letters},
  title         = {{Resource destroying maps}},
  year          = {2017},
  pages         = {060502},
  volume        = {118},
  archiveprefix = {arXiv},
  doi           = {10.1103/PhysRevLett.118.060502},
  eprint        = {1606.03723},
}

@Article{Wootters1998,
  author        = {Wootters, William K.},
  journal       = {Physical Review Letters},
  title         = {{Entanglement of formation of an arbitrary state of two qubits}},
  year          = {1998},
  pages         = {2245},
  volume        = {80},
  archiveprefix = {arXiv},
  doi           = {10.1103/PhysRevLett.80.2245},
  eprint        = {quant-ph/9709029},
}

@Article{Veitch2014,
  author        = {Veitch, Victor and Mousavian, S. A. Hamed and Gottesman, Daniel and Emerson, Joseph},
  journal       = {New Journal of Physics},
  title         = {{The resource theory of stabilizer quantum computation}},
  year          = {2014},
  pages         = {013009},
  volume        = {16},
  archiveprefix = {arXiv},
  doi           = {10.1088/1367-2630/16/1/013009},
  eprint        = {1307.7171},
}

@Article{Howard2017,
  author        = {Howard, Mark and Campbell, Earl T.},
  journal       = {Physical Review Letters},
  title         = {{Application of a resource theory for magic states to fault-tolerant quantum computing}},
  year          = {2017},
  pages         = {090501},
  volume        = {118},
  archiveprefix = {arXiv},
  doi           = {10.1103/PhysRevLett.118.090501},
  eprint        = {1609.07488},
}

@Article{Heinrich2019,
  author        = {Heinrich, Markus and Gross, David},
  journal       = {Quantum},
  title         = {{Robustness of magic and symmetries of the stabiliser polytope}},
  year          = {2019},
  pages         = {132},
  volume        = {3},
  archiveprefix = {arXiv},
  doi           = {10.22331/q-2019-04-08-132},
  eprint        = {1807.10296},
}

@Article{BravyiGosset2016,
  author  = {Bravyi, Sergey and Gosset, David},
  journal = {Physical Review Letters},
  title   = {{Improved classical simulation of quantum circuits dominated by Clifford gates}},
  year    = {2016},
  pages   = {250501},
  volume  = {116},
  doi     = {10.1103/PhysRevLett.116.250501},
}

@Article{DiVincenzo2002,
  author        = {DiVincenzo, David P. and Leung, Debbie W. and Terhal, Barbara M.},
  journal       = {IEEE Transactions on Information Theory},
  title         = {{Quantum data hiding}},
  year          = {2002},
  number        = {3},
  pages         = {580},
  volume        = {48},
  archiveprefix = {arXiv},
  doi           = {10.1109/18.985948},
  eprint        = {quant-ph/0103098},
}

@Article{Peres1996,
  author        = {Peres, Asher},
  journal       = {Physical Review Letters},
  title         = {{Separability criterion for density matrices}},
  year          = {1996},
  pages         = {1413},
  volume        = {77},
  archiveprefix = {arXiv},
  doi           = {10.1103/PhysRevLett.77.1413},
  eprint        = {quant-ph/9604005},
}

@Article{Horodecki1996,
  author        = {Horodecki, Micha{\l} and Horodecki, Pawe{\l} and Horodecki, Ryszard},
  journal       = {Physics Letters A},
  title         = {{Separability of mixed states: necessary and sufficient conditions}},
  year          = {1996},
  pages         = {1},
  volume        = {223},
  archiveprefix = {arXiv},
  doi           = {10.1016/S0375-9601(96)00706-2},
  eprint        = {quant-ph/9605038},
}

@Article{HorodeckiReview2009,
  author        = {Horodecki, Ryszard and Horodecki, Pawe{\l} and Horodecki, Micha{\l} and Horodecki, Karol},
  journal       = {Reviews of Modern Physics},
  title         = {{Quantum entanglement}},
  year          = {2009},
  pages         = {865},
  volume        = {81},
  archiveprefix = {arXiv},
  doi           = {10.1103/RevModPhysical81.865},
  eprint        = {quant-ph/0702225},
}

@Article{Sion1958,
  author  = {Sion, Maurice},
  journal = {Pacific Journal of Mathematics},
  title   = {{On general minimax theorems}},
  year    = {1958},
  number  = {1},
  pages   = {171},
  volume  = {8},
  doi     = {10.2140/pjm.1958.8.171},
}

@Article{Steiner2003,
  author        = {Steiner, Michael},
  journal       = {Physical Review A},
  title         = {{Generalized robustness of entanglement}},
  year          = {2003},
  pages         = {054305},
  volume        = {67},
  archiveprefix = {arXiv},
  doi           = {10.1103/PhysRevA.67.054305},
  eprint        = {quant-ph/0304009},
}

@Article{HarrowNielsen2003,
  author  = {Harrow, Aram W. and Nielsen, Michael A.},
  journal = {Physical Review A},
  title   = {{Robustness of quantum gates in the presence of noise}},
  year    = {2003},
  pages   = {012308},
  volume  = {68},
  doi     = {10.1103/PhysRevA.68.012308},
}

@Article{Regula2018,
  author        = {Regula, Bartosz},
  journal       = {Journal of Physics A: Mathematical and Theoretical},
  title         = {{Convex geometry of quantum resource quantification}},
  year          = {2018},
  pages         = {045303},
  volume        = {51},
  archiveprefix = {arXiv},
  doi           = {10.1088/1751-8121/aa9100},
  eprint        = {1707.06298},
}

@Article{BrandaoGour2015,
  author  = {Brand{\~a}o, Fernando G. S. L. and Gour, Gilad},
  journal = {Physical Review Letters},
  title   = {{Reversible framework for quantum resource theories}},
  year    = {2015},
  pages   = {070503},
  volume  = {115},
  doi     = {10.1103/PhysRevLett.115.070503},
}

@Misc{Hardy2001,
  author        = {Hardy, Lucien},
  title         = {{Quantum theory from five reasonable axioms}},
  year          = {2001},
  archiveprefix = {arXiv},
  eprint        = {quant-ph/0101012},
}

@Article{Barrett2007,
  author  = {Barrett, Jonathan},
  journal = {Physical Review A},
  title   = {{Information processing in generalized probabilistic theories}},
  year    = {2007},
  pages   = {032304},
  volume  = {75},
  doi     = {10.1103/PhysRevA.75.032304},
}

@Article{Plavala2023,
  author  = {Pl{\'a}vala, Martin},
  journal = {Physics Reports},
  title   = {{General probabilistic theories: An introduction}},
  year    = {2023},
  pages   = {1},
  volume  = {1033},
  doi     = {10.1016/j.physrep.2023.09.001},
}

@Article{Guerini2017,
  author  = {Guerini, Leonardo and Bavaresco, Jessica and Terra Cunha, Marcelo and Ac{\'i}n, Antonio},
  journal = {Journal of Mathematical Physics},
  title   = {{Operational framework for quantum measurement simulability}},
  year    = {2017},
  pages   = {092102},
  volume  = {58},
  doi     = {10.1063/1.4994303},
}

@Article{Baumgratz2014,
  author        = {Baumgratz, T. and Cramer, M. and Plenio, M. B.},
  journal       = {Physical Review Letters},
  title         = {{Quantifying coherence}},
  year          = {2014},
  pages         = {140401},
  volume        = {113},
  archiveprefix = {arXiv},
  doi           = {10.1103/PhysRevLett.113.140401},
  eprint        = {1311.0275},
}

@Article{Streltsov2017,
  author        = {Streltsov, Alexander and Adesso, Gerardo and Plenio, Martin B.},
  journal       = {Reviews of Modern Physics},
  title         = {{Colloquium: Quantum coherence as a resource}},
  year          = {2017},
  pages         = {041003},
  volume        = {89},
  archiveprefix = {arXiv},
  doi           = {10.1103/RevModPhysical89.041003},
  eprint        = {1609.02439},
}

@Article{GourSpekkens2008,
  author        = {Gour, Gilad and Spekkens, Robert W.},
  journal       = {New Journal of Physics},
  title         = {{The resource theory of quantum reference frames: manipulations and monotones}},
  year          = {2008},
  pages         = {033023},
  volume        = {10},
  archiveprefix = {arXiv},
  doi           = {10.1088/1367-2630/10/3/033023},
  eprint        = {0711.0043},
}

@Book{boyd2004convex,
  author    = {Boyd, Stephen and Vandenberghe, Lieven},
  publisher = {Cambridge University Press},
  title     = {Convex Optimization},
  year      = {2004},
  doi       = {10.1017/CBO9780511804441},
  place     = {Cambridge},
}

@Article{PhysRevLett.127.060402,
  author    = {Regula, Bartosz and Takagi, Ryuji},
  journal   = {Physical Review Letters},
  month     = {Aug},
  title     = {One-Shot Manipulation of Dynamical Quantum Resources},
  year      = {2021},
  pages     = {060402},
  volume    = {127},
  doi       = {10.1103/PhysRevLett.127.060402},
  issue     = {6},
  numpages  = {9},
  publisher = {American Physical Society},
  url       = {https://link.aps.org/doi/10.1103/PhysRevLett.127.060402},
}

@Article{Heimendahl2022,
  author        = {Heimendahl, Arne and Heinrich, Markus and Gross, David},
  journal       = {Journal of Mathematical Physics},
  title         = {The axiomatic and the operational approaches to resource theories of magic do not coincide},
  year          = {2022},
  number        = {11},
  pages         = {112201},
  volume        = {63},
  archiveprefix = {arXiv},
  doi           = {10.1063/5.0085774},
  eprint        = {2011.11651},
}

@Misc{ao2026geometricorigins,
  author        = {Jingsong Ao and Aby Philip and Alexander Streltsov},
  title         = {Resourcefulness without Resource: Geometric Origins and Robustness},
  year          = {2026},
  archiveprefix = {arXiv},
  eprint        = {2606.31516},
  primaryclass  = {quant-ph},
  url           = {https://arxiv.org/abs/2606.31516},
}

@Article{Helstrom1976,
  author    = {Helstrom, Carl W},
  journal   = {Journal of Statistical Physics},
  title     = {Quantum detection and estimation theory},
  year      = {1969},
  number    = {2},
  pages     = {231},
  volume    = {1},
  doi       = {10.1007/BF01007479},
  publisher = {Springer},
}

@Article{Holevo1973,
  author  = {Holevo, Alexander S.},
  journal = {Journal of Multivariate Analysis},
  title   = {Statistical Decision Theory for Quantum Systems},
  year    = {1973},
  number  = {4},
  pages   = {337--394},
  volume  = {3},
  doi     = {10.1016/0047-259X(73)90028-6},
}

@Misc{philip2025robustnessdatahiding,
  author        = {Aby Philip and Alexander Streltsov},
  title         = {Robustness of quantum data hiding against entangled catalysts and memory},
  year          = {2025},
  archiveprefix = {arXiv},
  eprint        = {2511.04408},
  primaryclass  = {quant-ph},
  url           = {https://arxiv.org/abs/2511.04408},
}

@Article{PhysRevLett.86.5807,
  author    = {Terhal, Barbara M. and DiVincenzo, David P. and Leung, Debbie W.},
  journal   = {Physical Review Letters},
  month     = {Jun},
  title     = {Hiding {B}its in {B}ell States},
  year      = {2001},
  pages     = {5807--5810},
  volume    = {86},
  doi       = {10.1103/PhysRevLett.86.5807},
  issue     = {25},
  numpages  = {0},
  publisher = {American Physical Society},
  url       = {https://link.aps.org/doi/10.1103/PhysRevLett.86.5807},
}

@Article{Hayden_2004,
  author    = {Hayden, Patrick and Leung, Debbie and Shor, Peter W. and Winter, Andreas},
  journal   = {Communications in Mathematical Physics},
  month     = jul,
  title     = {{Randomizing Quantum States: Constructions and Applications}},
  year      = {2004},
  issn      = {1432-0916},
  number    = {2},
  pages     = {371–391},
  volume    = {250},
  doi       = {10.1007/s00220-004-1087-6},
  publisher = {Springer Science and Business Media LLC},
  url       = {http://dx.doi.org/10.1007/s00220-004-1087-6},
}

@Article{Aubrun_2015,
  author    = {Aubrun, Guillaume and Lancien, C{\'e}cilia},
  journal   = {Quantum Information and Computation},
  month     = {apr},
  title     = {Locally restricted measurements on a multipartite quantum system: data hiding is generic},
  year      = {2015},
  pages     = {513--540},
  volume    = {15},
  doi       = {10.26421/qic15.5-6},
  issue     = {5-6},
  publisher = {Rinton Press},
  url       = {http://dx.doi.org/10.26421/QIC15.5-6},
}

@Article{SkrzypczykLinden2019,
  author    = {Skrzypczyk, Paul and Linden, Noah},
  journal   = {Physical Review Letters},
  month     = {Apr},
  title     = {{Robustness of measurement, discrimination games, and accessible information}},
  year      = {2019},
  pages     = {140403},
  volume    = {122},
  doi       = {10.1103/PhysRevLett.122.140403},
  issue     = {14},
  numpages  = {6},
  publisher = {American Physical Society},
  url       = {https://link.aps.org/doi/10.1103/PhysRevLett.122.140403},
}

@Article{EggelingWerner2002,
  author    = {Eggeling, Tilo and Werner, Reinhard F.},
  journal   = {Physical Review Letters},
  month     = {Aug},
  title     = {{Hiding classical data in multipartite quantum states}},
  year      = {2002},
  pages     = {097905},
  volume    = {89},
  doi       = {10.1103/PhysRevLett.89.097905},
  issue     = {9},
  numpages  = {4},
  publisher = {American Physical Society},
  url       = {https://link.aps.org/doi/10.1103/PhysRevLett.89.097905},
}

\end{document}